\documentclass[nofootinbib,twocolumn,pre,aps,10pt,floatfix]{revtex4-1}

\usepackage[T1]{fontenc}
\usepackage{lmodern}
\usepackage[utf8]{inputenc}
\usepackage{graphicx}
\usepackage{amsmath,amssymb,amsthm}
\usepackage{booktabs}
\usepackage{multirow}
\usepackage[binary-units=true]{siunitx}
\usepackage{expl3}
\usepackage[caption=false]{subfig}
\usepackage{enumitem}
\usepackage{algpseudocode}
\usepackage{newfloat}
\usepackage{xcolor}
\definecolor{link-color}{cmyk}{0.8,0.3,0.,0}
\usepackage[colorlinks=true,%
            linkcolor   = {link-color},%
            citecolor   = {link-color},%
            urlcolor    = {link-color}]{hyperref}

\renewcommand*\Call[2]{\textproc{#1}(#2)}

\usepackage{chngcntr}

\DeclareMathAlphabet{\pazocal}{OMS}{zplm}{m}{n}

\DeclareFloatingEnvironment[
    fileext=loa,
    listname=List of Algorithms,
    name=ALGORITHM,
    placement=tbhp,
]{algorithm}

\DeclareGraphicsExtensions{%
    .pdf,.PDF,%
    .png,.PNG}

\algdef{SE}[DOWHILE]{Do}{doWhile}{\algorithmicdo}[1]{\algorithmicwhile\ #1}%

\makeatletter
\newcommand*{\balancecolsandclearpage}{%
  \close@column@grid
  \cleardoublepage
  \twocolumngrid
}
\makeatother

\theoremstyle{definition}
\newtheorem{definition}{Definition}
\newtheorem{theorem}{Theorem}
\newtheorem{corollary}[theorem]{Corollary}
\newtheorem{lemma}[theorem]{Lemma}

\DeclareMathOperator*{\argmax}{arg\,max}
\DeclareMathOperator{\flatf}{flat}

\newcommand{\Hf}{\pazocal{H}}

\renewcommand{\P}{\pazocal{P}}

\newcommand{\T}{\pazocal{T}}

\usepackage{pifont}
\newcommand{\checksym}{\ding{51}}

\makeatletter
\newcommand{\simas}[1]{\mathbin{\overset{#1}{\kern\z@\sim}}}%
\newcommand{\nsimas}[1]{\mathbin{\overset{#1}{\kern\z@\nsim}}}%
\newsavebox{\mybox}\newsavebox{\mysim}
\newcommand{\simass}[1]{%
  \savebox{\mybox}{\hbox{\kern3pt$\scriptstyle#1$\kern3pt}}%
  \savebox{\mysim}{\hbox{$\sim$}}%
  \mathbin{\overset{#1}{\kern\z@\resizebox{\wd\mybox}{\ht\mysim}{$\sim$}}}%
}
\makeatother

\graphicspath{ {figures/} }


\begin{document}

\title{From Louvain to Leiden: guaranteeing well-connected communities}
\author{V.A. Traag}
\email{v.a.traag@cwts.leidenuniv.nl}
\affiliation{Centre for Science and Technology Studies, Leiden University, the Netherlands}
\author{L. Waltman}
\affiliation{Centre for Science and Technology Studies, Leiden University, the Netherlands}
\author{N.J. van Eck}
\affiliation{Centre for Science and Technology Studies, Leiden University, the Netherlands}

\date{\today}

\begin{abstract}
  Community detection is often used to understand the structure of large and complex networks.
  One of the most popular algorithms for uncovering community structure is the so-called Louvain algorithm.
  We show that this algorithm has a major defect that largely went unnoticed until now: the Louvain algorithm may yield arbitrarily badly connected communities.
  In the worst case, communities may even be disconnected, especially when running the algorithm iteratively.
  In our experimental analysis, we observe that up to $25\%$ of the communities are badly connected and up to $16\%$ are disconnected.
  To address this problem, we introduce the Leiden algorithm.
  We prove that the Leiden algorithm yields communities that are guaranteed to be connected.
  In addition, we prove that, when the Leiden algorithm is applied iteratively, it converges to a partition in which all subsets of all communities are locally optimally assigned.
  Furthermore, by relying on a fast local move approach, the Leiden algorithm runs faster than the Louvain algorithm.
  We demonstrate the performance of the Leiden algorithm for several benchmark and real-world networks.
  We find that the Leiden algorithm is faster than the Louvain algorithm and uncovers better partitions, in addition to providing explicit guarantees.
\end{abstract}

\maketitle

\section{Introduction}

\noindent In many complex networks, nodes cluster and form relatively dense groups---often called communities~\cite{Fortunato2010,Porter2009}.
Such a modular structure is usually not known beforehand.
Detecting communities in a network is therefore an important problem.
One of the best-known methods for community detection is called modularity~\cite{Newman2004Finding}.
This method tries to maximise the difference between the actual number of edges in a community and the expected number of such edges.
We denote by $e_c$ the actual number of edges in community $c$.
The expected number of edges can be expressed as $\frac{K_c^2}{2m}$, where $K_c$ is the sum of the degrees of the nodes in community $c$ and $m$ is the total number of edges in the network.
This way of defining the expected number of edges is based on the so-called configuration model.
Modularity is given by
\begin{equation}
  \Hf = \frac{1}{2m} \sum_c \left(e_c - \gamma \frac{K_c^2}{2m} \right),
\end{equation}
where $\gamma > 0$ is a resolution parameter~\cite{Reichardt2006Statistical}.
Higher resolutions lead to more communities, while lower resolutions lead to fewer communities.

Optimising modularity is NP-hard~\cite{Brandes}, and consequentially many heuristic algorithms have been proposed, such as hierarchical agglomeration~\cite{Clauset2004}, extremal optimisation~\cite{Duch2005}, simulated annealing~\cite{Reichardt2006Statistical,Guimera2005Functional} and spectral~\cite{Newman2006Finding} algorithms.
One of the most popular algorithms to optimise modularity is the so-called Louvain algorithm~\cite{Blondel2008}, named after the location of its authors.
It was found to be one of the fastest and best performing algorithms in comparative analyses~\cite{Lancichinetti2009,Yang2016}, and it is one of the most-cited works in the community detection literature.

Although originally defined for modularity, the Louvain algorithm can also be used to optimise other quality functions.
An alternative quality function is the Constant Potts Model (CPM)~\cite{Traag2011}, which overcomes some limitations of modularity.
CPM is defined as
\begin{equation}
  \Hf = \sum_c \left[e_c - \gamma \binom{n_c}{2} \right],
  \label{eq:CPM_simple}
\end{equation}
where $n_c$ is the number of nodes in community $c$.
The interpretation of the resolution parameter $\gamma$ is quite straightforward.
The parameter functions as a sort of threshold: communities should have a density of at least $\gamma$, while the density between communities should be lower than $\gamma$.
Higher resolutions lead to more communities and lower resolutions lead to fewer communities, similarly to the resolution parameter for modularity.

In this paper, we show that the Louvain algorithm has a major problem, for both modularity and CPM.
The algorithm may yield arbitrarily badly connected communities, over and above the well-known issue of the resolution limit~\cite{Fortunato:2007p183} (Section~\ref{sec:disconnected}).
Communities may even be internally disconnected.
To address this important shortcoming, we introduce a new algorithm that is faster, finds better partitions and provides explicit guarantees and bounds (Section~\ref{sec:leiden}).
The new algorithm integrates several earlier improvements, incorporating a combination of smart local move~\cite{Waltman2013}, fast local move~\cite{Ozaki,Bae2014} and random neighbour move~\cite{Traag2015a}.
We prove that the new algorithm is guaranteed to produce partitions in which all communities are internally connected.
In addition, we prove that the algorithm converges to an asymptotically stable partition in which all subsets of all communities are locally optimally assigned.
The quality of such an asymptotically stable partition provides an upper bound on the quality of an optimal partition.
Finally, we demonstrate the excellent performance of the algorithm for several benchmark and real-world networks (Section~\ref{sec:analysis}).
To ensure readability of the paper to the broadest possible audience, we have chosen to relegate all technical details to appendices.
The main ideas of our algorithm are explained in an intuitive way in the main text of the paper.
We name our algorithm the \emph{Leiden algorithm}, after the location of its authors.

\section{Louvain algorithm}

\begin{figure}[tb]
  \begin{center}
    \includegraphics{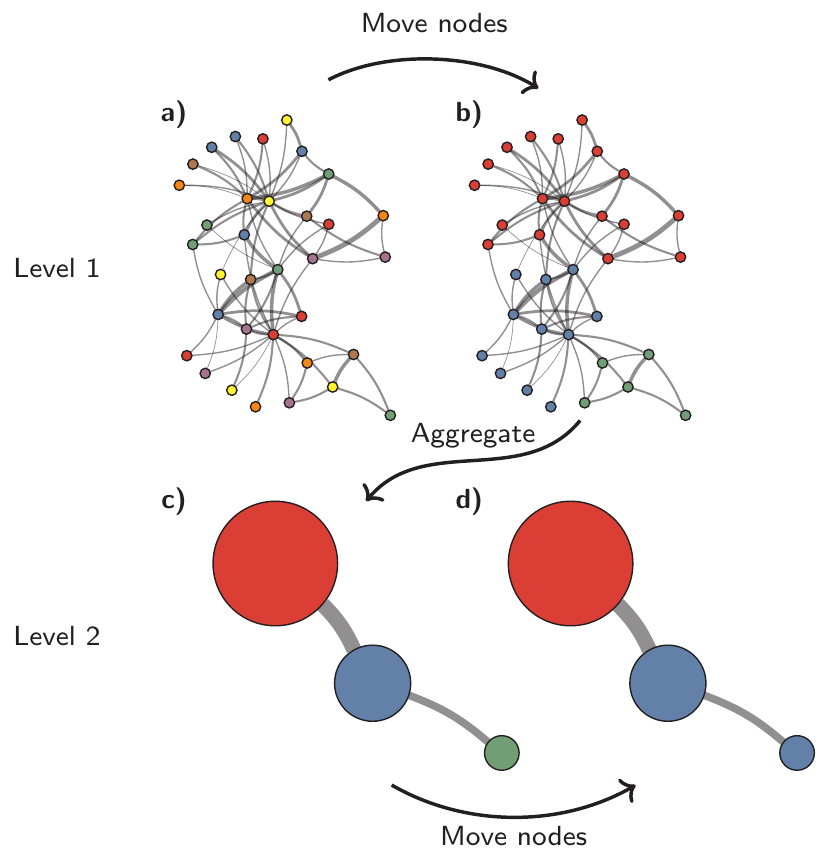}
  \end{center}
  \caption{\textbf{Louvain algorithm}.
    The Louvain algorithm starts from a singleton partition in which each node is in its own community (a).
    The algorithm moves individual nodes from one community to another to find a partition (b).
    Based on this partition, an aggregate network is created (c).
    The algorithm then moves individual nodes in the aggregate network (d).
    These steps are repeated until the quality cannot be increased further.
    }
  \label{fig:louvain_illustration}
\end{figure}

\noindent The Louvain algorithm~\cite{Blondel2008} is very simple and elegant.
The algorithm optimises a quality function such as modularity or CPM in two elementary phases: (1) local moving of nodes; and (2) aggregation of the network.
In the local moving phase, individual nodes are moved to the community that yields the largest increase in the quality function.
In the aggregation phase, an aggregate network is created based on the partition obtained in the local moving phase.
Each community in this partition becomes a node in the aggregate network.
The two phases are repeated until the quality function cannot be increased further.
The Louvain algorithm is illustrated in Fig.~\ref{fig:louvain_illustration} and summarised in pseudo-code in Algorithm~\ref{algo:louvain} in Appendix~\ref{sec:code_notation}.

Usually, the Louvain algorithm starts from a singleton partition, in which each node is in its own community.
However, it is also possible to start the algorithm from a different partition~\cite{Waltman2013}.
In particular, in an attempt to find better partitions, multiple consecutive iterations of the algorithm can be performed, using the partition identified in one iteration as starting point for the next iteration.

\subsection{Badly connected communities}
\label{sec:disconnected}

\begin{figure}[bt]
  \begin{center}
    \includegraphics{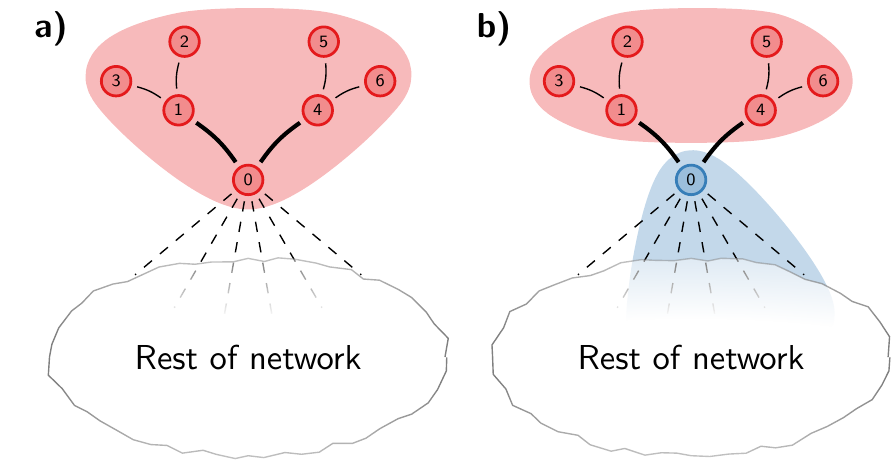}
  \end{center}
  \caption{\textbf{Disconnected community.}
    Consider the partition shown in (a).
    When node 0 is moved to a different community, the red community becomes internally disconnected, as shown in (b).
    However, nodes 1--6 are still locally optimally assigned, and therefore these nodes will stay in the red community.
    }
  \label{fig:disconnected_community}
\end{figure}

We now show that the Louvain algorithm may find arbitrarily badly connected communities.
In particular, we show that Louvain may identify communities that are internally disconnected.
That is, one part of such an internally disconnected community can reach another part only through a path going outside the community.
Importantly, the problem of disconnected communities is not just a theoretical curiosity.
As we will demonstrate in Section~\ref{sec:analysis}, the problem occurs frequently in practice when using the Louvain algorithm.
Perhaps surprisingly, iterating the algorithm aggravates the problem, even though it does increase the quality function.

In the Louvain algorithm, a node may be moved to a different community while it may have acted as a bridge between different components of its old community.
Removing such a node from its old community disconnects the old community.
One may expect that other nodes in the old community will then also be moved to other communities.
However, this is not necessarily the case, as the other nodes may still be sufficiently strongly connected to their community, despite the fact that the community has become disconnected.

To elucidate the problem, we consider the example illustrated in Fig.~\ref{fig:disconnected_community}.
The numerical details of the example can be found in Appendix~\ref{sec:disconnected_example}.
The thick edges in Fig.~\ref{fig:disconnected_community} represent stronger connections, while the other edges represent weaker connections.
At some point, the Louvain algorithm may end up in the community structure shown in Fig.~\ref{fig:disconnected_community}(a).
Nodes 0--6 are in the same community.
Nodes 1--6 have connections only within this community, whereas node 0 also has many external connections.
The algorithm continues to move nodes in the rest of the network.
At some point, node 0 is considered for moving.
When a sufficient number of neighbours of node 0 have formed a community in the rest of the network, it may be optimal to move node 0 to this community, thus creating the situation depicted in Fig.~\ref{fig:disconnected_community}(b).
In this new situation, nodes 2, 3, 5 and 6 have only internal connections.
These nodes are therefore optimally assigned to their current community.
On the other hand, after node 0 has been moved to a different community, nodes 1 and 4 have not only internal but also external connections.
Nevertheless, depending on the relative strengths of the different connections, these nodes may still be optimally assigned to their current community.
In that case, nodes 1--6 are all locally optimally assigned, despite the fact that their community has become disconnected.
Clearly, it would be better to split up the community.
Nodes 1--3 should form a community and nodes 4--6 should form another community.
However, the Louvain algorithm does not consider this possibility, since it considers only individual node movements.
Moreover, when no more nodes can be moved, the algorithm will aggregate the network.
When a disconnected community has become a node in an aggregate network, there are no more possibilities to split up the community.
Hence, the community remains disconnected, unless it is merged with another community that happens to act as a bridge.

Obviously, this is a worst case example, showing that disconnected communities may be identified by the Louvain algorithm.
More subtle problems may occur as well, causing Louvain to find communities that are connected, but only in a very weak sense.
Hence, in general, Louvain may find arbitrarily badly connected communities.

This problem is different from the well-known issue of the resolution limit of modularity~\cite{Fortunato:2007p183}.
Due to the resolution limit, modularity may cause smaller communities to be clustered into larger communities.
In other words, modularity may ``hide'' smaller communities and may yield communities containing significant substructure.
CPM does not suffer from this issue~\cite{Traag2011}.
Nevertheless, when CPM is used as the quality function, the Louvain algorithm may still find arbitrarily badly connected communities.
Hence, the problem of Louvain outlined above is independent from the issue of the resolution limit.
In the case of modularity, communities may have significant substructure both because of the resolution limit and because of the shortcomings of Louvain.

In fact, although it may seem that the Louvain algorithm does a good job at finding high quality partitions, in its standard form the algorithm provides only one guarantee: the algorithm yields partitions for which it is guaranteed that no communities can be merged.
In other words, communities are guaranteed to be well separated.
Somewhat stronger guarantees can be obtained by iterating the algorithm, using the partition obtained in one iteration of the algorithm as starting point for the next iteration.
When iterating Louvain, the quality of the partitions will keep increasing until the algorithm is unable to make any further improvements.
At this point, it is guaranteed that each individual node is optimally assigned.
In this iterative scheme, Louvain provides two guarantees: (1) no communities can be merged and (2) no nodes can be moved.

Contrary to what might be expected, iterating the Louvain algorithm aggravates the problem of badly connected communities, as we will also see in Section~\ref{sec:analysis}.
This is not too difficult to explain.
After the first iteration of the Louvain algorithm, some partition has been obtained.
In the first step of the next iteration, Louvain will again move individual nodes in the network.
Some of these nodes may very well act as bridges, similarly to node $0$ in the above example.
By moving these nodes, Louvain creates badly connected communities.
Moreover, Louvain has no mechanism for fixing these communities.
Iterating the Louvain algorithm can therefore be seen as a double-edged sword: it improves the partition in some way, but degrades it in another way.

The problem of disconnected communities has been observed before in the context of the label propagation algorithm~\cite{Raghavan2007}.
However, so far this problem has never been studied for the Louvain algorithm.
Moreover, the deeper significance of the problem was not recognised: disconnected communities are merely the most extreme manifestation of the problem of arbitrarily badly connected communities.
Trying to fix the problem by simply considering the connected components of communities~\cite{Luecken2016,Wolf2018,Raghavan2007} is unsatisfactory because it addresses only the most extreme case and does not resolve the more fundamental problem.
We therefore require a more principled solution, which we will introduce in the next section.

\begin{figure*}[tb]
  \begin{center}
    \includegraphics[width=\textwidth]{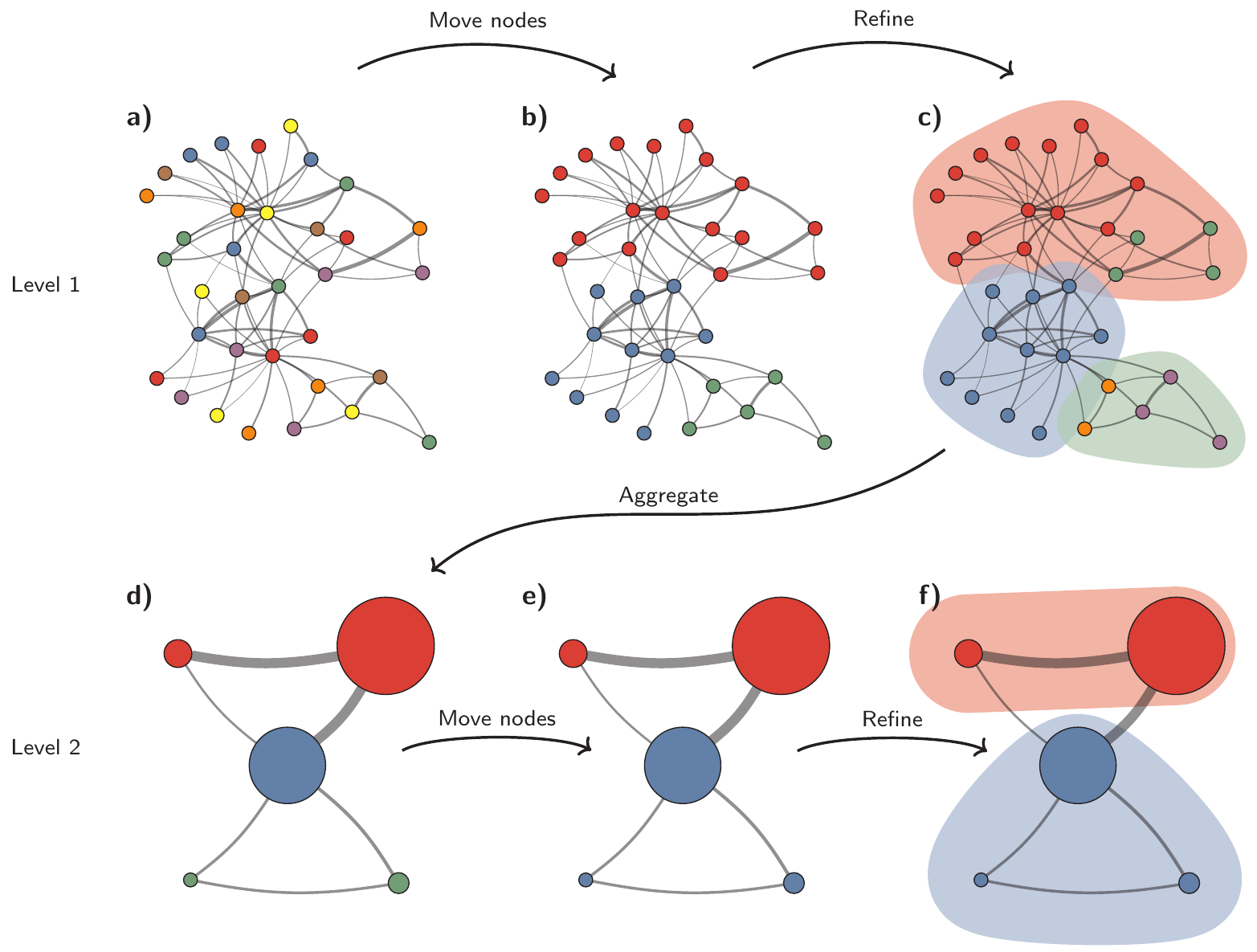}
  \end{center}
  \caption{\textbf{Leiden algorithm}.
    The Leiden algorithm starts from a singleton partition (a).
    The algorithm moves individual nodes from one community to another to find a partition (b), which is then refined (c).
    An aggregate network (d) is created based on the refined partition, using the non-refined partition to create an initial partition for the aggregate network.
    For example, the red community in (b) is refined into two subcommunities in (c), which after aggregation become two separate nodes in (d), both belonging to the same community.
    The algorithm then moves individual nodes in the aggregate network (e).
	In this case, refinement does not change the partition (f).
    These steps are repeated until no further improvements can be made.
    }
  \label{fig:leiden_illustration}
\end{figure*}

\section{Leiden algorithm}
\label{sec:leiden}

\noindent We here introduce the Leiden algorithm, which guarantees that communities are well connected.
The Leiden algorithm is partly based on the previously introduced smart local move algorithm~\cite{Waltman2013}, which itself can be seen as an improvement of the Louvain algorithm.
The Leiden algorithm also takes advantage of the idea of speeding up the local moving of nodes~\cite{Ozaki,Bae2014} and the idea of moving nodes to random neighbours~\cite{Traag2015a}.
We consider these ideas to represent the most promising directions in which the Louvain algorithm can be improved, even though we recognise that other improvements have been suggested as well~\cite{Rotta2011}.
The Leiden algorithm consists of three phases: (1) local moving of nodes, (2) refinement of the partition and (3) aggregation of the network based on the refined partition, using the non-refined partition to create an initial partition for the aggregate network.
The Leiden algorithm is considerably more complex than the Louvain algorithm.
Fig.~\ref{fig:leiden_illustration} provides an illustration of the algorithm.
The algorithm is described in pseudo-code in Algorithm~\ref{algo:leiden} in Appendix~\ref{sec:code_notation}.

In the Louvain algorithm, an aggregate network is created based on the partition $\P$ resulting from the local moving phase.
The idea of the refinement phase in the Leiden algorithm is to identify a partition $\P_\text{refined}$ that is a refinement of $\P$.
Communities in $\P$ may be split into multiple subcommunities in $\P_\text{refined}$.
The aggregate network is created based on the partition $\P_\text{refined}$.
However, the initial partition for the aggregate network is based on $\P$, just like in the Louvain algorithm.
By creating the aggregate network based on $\P_\text{refined}$ rather than $\P$, the Leiden algorithm has more room for identifying high-quality partitions.
In fact, by implementing the refinement phase in the right way, several attractive guarantees can be given for partitions produced by the Leiden algorithm.

The refined partition $\P_\text{refined}$ is obtained as follows.
Initially, $\P_\text{refined}$ is set to a singleton partition, in which each node is in its own community.
The algorithm then locally merges nodes in $\P_\text{refined}$: nodes that are on their own in a community in $\P_\text{refined}$ can be merged with a different community.
Importantly, mergers are performed only within each community of the partition $\P$.
In addition, a node is merged with a community in $\P_\text{refined}$ only if both are sufficiently well connected to their community in $\P$.
After the refinement phase is concluded, communities in $\P$ often will have been split into multiple communities in $\P_\text{refined}$, but not always.

In the refinement phase, nodes are not necessarily greedily merged with the community that yields the largest increase in the quality function.
Instead, a node may be merged with any community for which the quality function increases.
The community with which a node is merged is selected randomly (similar to~\cite{Traag2015a}).
The larger the increase in the quality function, the more likely a community is to be selected.
The degree of randomness in the selection of a community is determined by a parameter $\theta > 0$.
Randomness in the selection of a community allows the partition space to be explored more broadly.
Node mergers that cause the quality function to decrease are not considered.
This contrasts with optimisation algorithms such as simulated annealing, which do allow the quality function to decrease~\cite{Guimera2005Functional,Reichardt2006Statistical}.
Such algorithms are rather slow, making them ineffective for large networks.
Excluding node mergers that decrease the quality function makes the refinement phase more efficient.
As we prove in Appendix~\ref{sec:nondecreasing_move_sequences}, even when node mergers that decrease the quality function are excluded, the optimal partition of a set of nodes can still be uncovered.
This is not the case when nodes are greedily merged with the community that yields the largest increase in the quality function.
In that case, some optimal partitions cannot be found, as we show in Appendix~\ref{sec:greedy_move_sequences}.

Another important difference between the Leiden algorithm and the Louvain algorithm is the implementation of the local moving phase.
Unlike the Louvain algorithm, the Leiden algorithm uses a fast local move procedure in this phase.
Louvain keeps visiting all nodes in a network until there are no more node movements that increase the quality function.
In doing so, Louvain keeps visiting nodes that cannot be moved to a different community.
In the fast local move procedure in the Leiden algorithm, only nodes whose neighbourhood has changed are visited.
This is similar to ideas proposed recently as ``pruning''~\cite{Ozaki} and in a slightly different form as ``prioritisation''~\cite{Bae2014}.
The fast local move procedure can be summarised as follows.
We start by initialising a queue with all nodes in the network.
The nodes are added to the queue in a random order.
We then remove the first node from the front of the queue and we determine whether the quality function can be increased by moving this node from its current community to a different one.
If we move the node to a different community, we add to the rear of the queue all neighbours of the node that do not belong to the node's new community and that are not yet in the queue.
We keep removing nodes from the front of the queue, possibly moving these nodes to a different community.
This continues until the queue is empty.
For a full specification of the fast local move procedure, we refer to the pseudo-code of the Leiden algorithm in Algorithm~\ref{algo:leiden} in Appendix~\ref{sec:code_notation}.
Using the fast local move procedure, the first visit to all nodes in a network in the Leiden algorithm is the same as in the Louvain algorithm.
However, after all nodes have been visited once, Leiden visits only nodes whose neighbourhood has changed, whereas Louvain keeps visiting all nodes in the network.
In this way, Leiden implements the local moving phase more efficiently than Louvain.

\begin{table}[tb]
  \caption{Overview of the guarantees provided by the Louvain algorithm and the Leiden algorithm.}
  \label{tbl:guarantees}
  \begin{tabular}{lrcc}
    \toprule
    &                          & Louvain   & Leiden    \\
    \midrule
    \multirow{2}{2cm}{Each iteration}
    & $\gamma$-separation      & \checksym & \checksym \\
    & $\gamma$-connectivity    &           & \checksym \\
    \midrule
    \multirow{2}{2cm}{Stable iteration}
    & Node optimality          & \checksym & \checksym \\
    & Subpartition $\gamma$-density &      & \checksym \\
    \midrule
    \multirow{2}{2cm}{Asymptotic}
    & Uniform $\gamma$-density &           & \checksym \\
    & Subset optimality        &           & \checksym \\
    \bottomrule
  \end{tabular}
\end{table}

\subsection{Guarantees}

We now consider the guarantees provided by the Leiden algorithm.
The algorithm is run iteratively, using the partition identified in one iteration as starting point for the next iteration.
We can guarantee a number of properties of the partitions found by the Leiden algorithm at various stages of the iterative process.
Below we offer an intuitive explanation of these properties.
We provide the full definitions of the properties as well as the mathematical proofs in Appendix~\ref{sec:guarantees}.

After each iteration of the Leiden algorithm, it is guaranteed that:
\begin{enumerate}
  \item All communities are $\gamma$-separated.
  \item All communities are $\gamma$-connected.
\end{enumerate}
In these properties, $\gamma$ refers to the resolution parameter in the quality function that is optimised, which can be either modularity or CPM.
The property of $\gamma$-separation is also guaranteed by the Louvain algorithm.
It states that there are no communities that can be merged.
The property of $\gamma$-connectivity is a slightly stronger variant of ordinary connectivity.
As discussed in Section~\ref{sec:disconnected}, the Louvain algorithm does not guarantee connectivity.
It therefore does not guarantee $\gamma$-connectivity either.

An iteration of the Leiden algorithm in which the partition does not change is called a stable iteration.
After a stable iteration of the Leiden algorithm, it is guaranteed that:
\begin{enumerate}[resume]
  \item All nodes are locally optimally assigned.
  \item All communities are subpartition $\gamma$-dense.
\end{enumerate}
Node optimality is also guaranteed after a stable iteration of the Louvain algorithm.
It means that there are no individual nodes that can be moved to a different community.
Subpartition $\gamma$-density is not guaranteed by the Louvain algorithm.
A community is subpartition $\gamma$-dense if it can be partitioned into two parts such that: (1) the two parts are well connected to each other; (2) neither part can be separated from its community; and (3) each part is also subpartition $\gamma$-dense itself.
Subpartition $\gamma$-density does not imply that individual nodes are locally optimally assigned.
It only implies that individual nodes are well connected to their community.

In the case of the Louvain algorithm, after a stable iteration, all subsequent iterations will be stable as well.
Hence, no further improvements can be made after a stable iteration of the Louvain algorithm.
This contrasts with the Leiden algorithm.
After a stable iteration of the Leiden algorithm, the algorithm may still be able to make further improvements in later iterations.
In fact, when we keep iterating the Leiden algorithm, it will converge to a partition for which it is guaranteed that:
\begin{enumerate}[resume]
  \item All communities are uniformly $\gamma$-dense.
  \item All communities are subset optimal.
\end{enumerate}
A community is uniformly $\gamma$-dense if there are no subsets of the community that can be separated from the community.
Uniform $\gamma$-density means that no matter how a community is partitioned into two parts, the two parts will always be well connected to each other.
Furthermore, if all communities in a partition are uniformly $\gamma$-dense, the quality of the partition is not too far from optimal, as shown in Appendix~\ref{sec:bound_on_optimality}.
A community is subset optimal if all subsets of the community are locally optimally assigned.
That is, no subset can be moved to a different community.
Subset optimality is the strongest guarantee that is provided by the Leiden algorithm.
It implies uniform $\gamma$-density and all the other above-mentioned properties.

An overview of the various guarantees is presented in Table~\ref{tbl:guarantees}.

\section{Experimental analysis}
\label{sec:analysis}

\begin{table}[tb]
\caption{Overview of the empirical networks and of the maximal modularity after $10$ replications of $10$ iterations each, both for the Louvain and for the Leiden algorithm.}
\label{tbl:real_networks}
\begin{tabular}{lrrrr}
  \toprule
                  &                &        & \multicolumn{2}{c}{Max. modularity}\\
  \cmidrule{4-5}
                  & Nodes          & Degree & Louvain  & Leiden                  \\
  \midrule
  DBLP\footnotemark[1]
                  &     $317\,080$ &  $6.6$ & $0.8262$ & $0.8387$                \\
  Amazon\footnotemark[1]
                  &     $334\,863$ &  $5.6$ & $0.9301$ & $0.9341$                \\
  IMDB\footnotemark[2]
                  &     $374\,511$ & $80.2$ & $0.7062$ & $0.7069$                \\
  Live Journal\footnotemark[1]
                  &  $3\,997\,962$ & $17.4$ & $0.7653$ & $0.7739$                \\
  Web of Science\footnotemark[3]
                  &  $9\,811\,130$ & $21.2$ & $0.7911$ & $0.7951$                \\
  Web UK\footnotemark[4]
                  & $39\,252\,879$ & $39.8$ & $0.9796$ & $0.9801$                \\
  \bottomrule
\end{tabular}
\footnotetext[1]{\url{https://snap.stanford.edu/data/}}
\footnotetext[2]{\url{https://sparse.tamu.edu/Barabasi/NotreDame_actors}}
\footnotetext[3]{Data cannot be shared due to license restrictions.}
\footnotetext[4]{\url{http://law.di.unimi.it/webdata/uk-2005/}}
\end{table}

In the previous section, we showed that the Leiden algorithm guarantees a number of properties of the partitions uncovered at different stages of the algorithm.
We also suggested that the Leiden algorithm is faster than the Louvain algorithm, because of the fast local move approach.
In this section, we analyse and compare the performance of the two algorithms in practice\footnote{We implemented both algorithms in Java, available from \href{https://github.com/CWTSLeiden/networkanalysis}{github.com/CWTSLeiden/networkanalysis} and deposited at Zenodo~\cite{code}.
Additionally, we implemented a Python package, available from \href{https://github.com/vtraag/leidenalg}{github.com/vtraag/leidenalg} and deposited at Zenodo~\cite{codePython}.}.
All experiments were run on a computer with 64 Intel Xeon E5-4667v3 2GHz CPUs and 1TB internal memory.
In all experiments reported here, we used a value of $0.01$ for the parameter $\theta$ that determines the degree of randomness in the refinement phase of the Leiden algorithm.
However, values of $\theta$ within a range of roughly $[0.0005, 0.1]$ all provide reasonable results, thus allowing for some, but not too much randomness.
We use six empirical networks in our analysis.
These are the same networks that were also studied in an earlier paper introducing the smart local move algorithm~\cite{Waltman2013}.
Table~\ref{tbl:real_networks} provides an overview of the six networks.
First, we show that the Louvain algorithm finds disconnected communities, and more generally, badly connected communities in the empirical networks.
Second, to study the scaling of the Louvain and the Leiden algorithm, we use benchmark networks, allowing us to compare the algorithms in terms of both computational time and quality of the partitions.
Finally, we compare the performance of the algorithms on the empirical networks.
We find that the Leiden algorithm commonly finds partitions of higher quality in less time.
The difference in computational time is especially pronounced for larger networks, with Leiden being up to $20$ times faster than Louvain in empirical networks.

\subsection{Badly connected communities}

We study the problem of badly connected communities when using the Louvain algorithm for several empirical networks.
For each community in a partition that was uncovered by the Louvain algorithm, we determined whether it is internally connected or not.
In addition, to analyse whether a community is badly connected, we ran the Leiden algorithm on the subnetwork consisting of all nodes belonging to the community.\footnote{
  We ensured that modularity optimisation for the subnetwork was fully consistent with modularity optimisation for the whole network~\cite{Traag2011}.}
The Leiden algorithm was run until a stable iteration was obtained.
When the Leiden algorithm found that a community could be split into multiple subcommunities, we counted the community as badly connected.
Note that if Leiden finds subcommunities, splitting up the community is guaranteed to increase modularity.
Conversely, if Leiden does not find subcommunities, there is no guarantee that modularity cannot be increased by splitting up the community.
Hence, by counting the number of communities that have been split up, we obtained a lower bound on the number of communities that are badly connected.
The count of badly connected communities also included disconnected communities.
For each network, we repeated the experiment $10$ times.
We used modularity with a resolution parameter of $\gamma = 1$ for the experiments.

\begin{figure}[tb]
  \centering
  \includegraphics{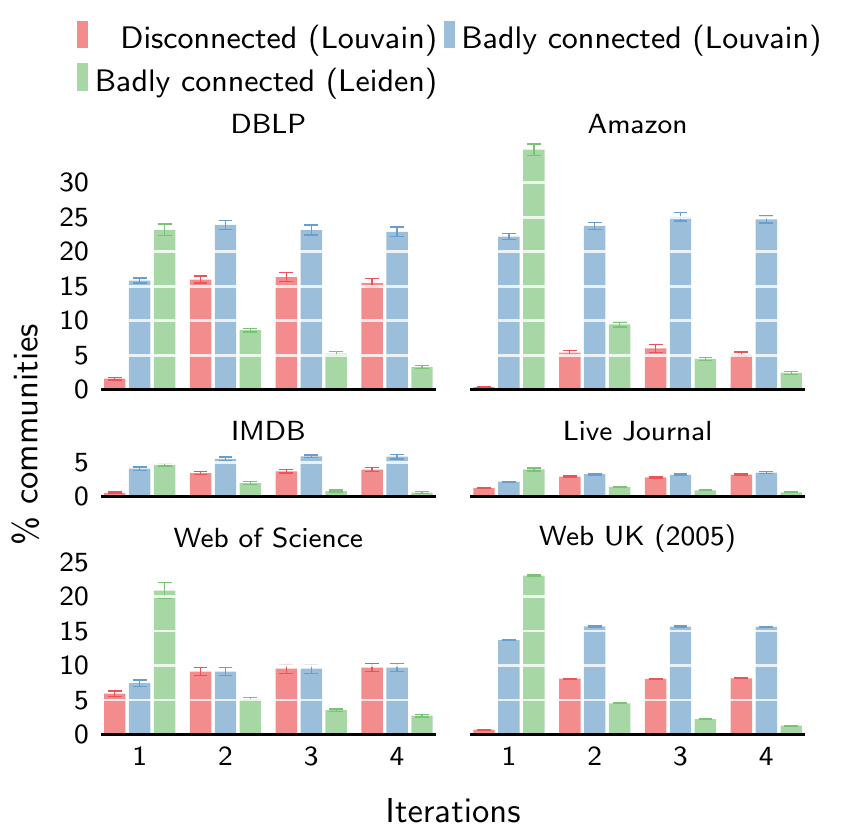}
  \caption{
    \textbf{Badly connected communities}.
    Percentage of communities found by the Louvain algorithm that are either disconnected or badly connected compared to percentage of badly connected communities found by the Leiden algorithm.
	Note that communities found by the Leiden algorithm are guaranteed to be connected.
  }
  \label{fig:subcluster}
\end{figure}

As can be seen in Fig.~\ref{fig:subcluster}, in the first iteration of the Louvain algorithm, the percentage of badly connected communities can be quite high.
For the Amazon, DBLP and Web UK networks, Louvain yields on average respectively $23\%$, $16\%$ and $14\%$ badly connected communities.
The percentage of disconnected communities is more limited, usually around $1\%$.
However, in the case of the Web of Science network, more than $5\%$ of the communities are disconnected in the first iteration.

Later iterations of the Louvain algorithm only aggravate the problem of disconnected communities, even though the quality function (i.e. modularity) increases.
The second iteration of Louvain shows a large increase in the percentage of disconnected communities.
In subsequent iterations, the percentage of disconnected communities remains fairly stable.
The increase in the percentage of disconnected communities is relatively limited for the Live Journal and Web of Science networks.
Other networks show an almost tenfold increase in the percentage of disconnected communities.
The percentage of disconnected communities even jumps to $16\%$ for the DBLP network.
The percentage of badly connected communities is less affected by the number of iterations of the Louvain algorithm.
Presumably, many of the badly connected communities in the first iteration of Louvain become disconnected in the second iteration.
Indeed, the percentage of disconnected communities becomes more comparable to the percentage of badly connected communities in later iterations.
Nonetheless, some networks still show large differences.
For example, after four iterations, the Web UK network has $8\%$ disconnected communities, but twice as many badly connected communities.
Even worse, the Amazon network has $5\%$ disconnected communities, but $25\%$ badly connected communities.

The above results shows that the problem of disconnected and badly connected communities is quite pervasive in practice.
Because the percentage of disconnected communities in the first iteration of the Louvain algorithm usually seems to be relatively low, the problem may have escaped attention from users of the algorithm.
However, focussing only on disconnected communities masks the more fundamental issue: Louvain finds arbitrarily badly connected communities.
The high percentage of badly connected communities attests to this.
Besides being pervasive, the problem is also sizeable.
In the worst case, almost a quarter of the communities are badly connected.
This may have serious consequences for analyses based on the resulting partitions.
For example, nodes in a community in biological or neurological networks are often assumed to share similar functions or behaviour~\cite{Bullmore2009}.
However, if communities are badly connected, this may lead to incorrect attributions of shared functionality.
Similarly, in citation networks, such as the Web of Science network, nodes in a community are usually considered to share a common topic~\cite{Waltman2012,Klavans2017}.
Again, if communities are badly connected, this may lead to incorrect inferences of topics, which will affect bibliometric analyses relying on the inferred topics.
In short, the problem of badly connected communities has important practical consequences.

The Leiden algorithm has been specifically designed to address the problem of badly connected communities.
Fig.~\ref{fig:subcluster} shows how well it does compared to the Louvain algorithm.
The Leiden algorithm guarantees all communities to be connected, but it may yield badly connected communities.
In terms of the percentage of badly connected communities in the first iteration, Leiden performs even worse than Louvain, as can be seen in Fig.~\ref{fig:subcluster}.
Crucially, however, the percentage of badly connected communities decreases with each iteration of the Leiden algorithm.
Starting from the second iteration, Leiden outperformed Louvain in terms of the percentage of badly connected communities.
In fact, if we keep iterating the Leiden algorithm, it will converge to a partition without any badly connected communities, as discussed in Section~\ref{sec:leiden}.
Hence, the Leiden algorithm effectively addresses the problem of badly connected communities.

\begin{figure*}[tb]
  \begin{center}
    \includegraphics{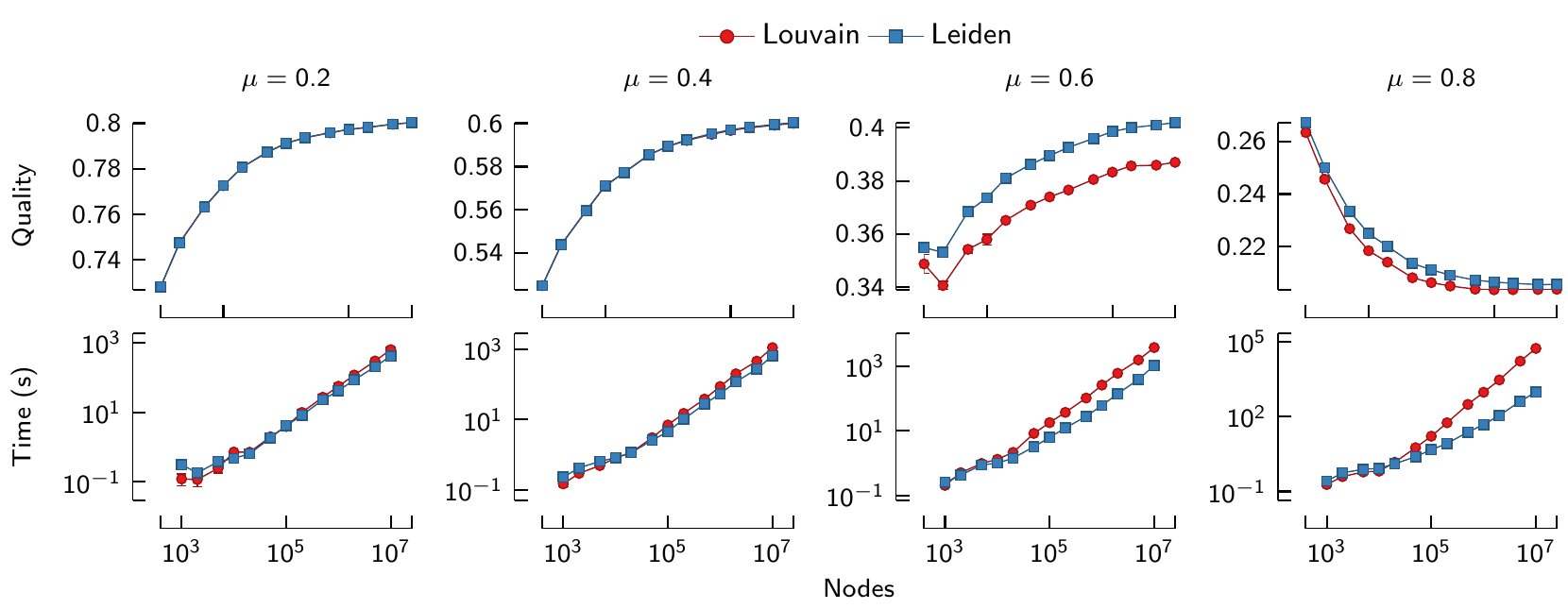}
  \end{center}
  \caption{
    \textbf{Scaling of benchmark results for network size}.
    Speed and quality of the Louvain and the Leiden algorithm for benchmark networks of increasing size (two iterations).
    For larger networks and higher values of $\mu$, Louvain is much slower than Leiden.
	For higher values of $\mu$, Leiden finds better partitions than Louvain.
  }
  \label{fig:benchmark_nodes}
\end{figure*}

\begin{figure*}[tb]
  \begin{center}
    \includegraphics{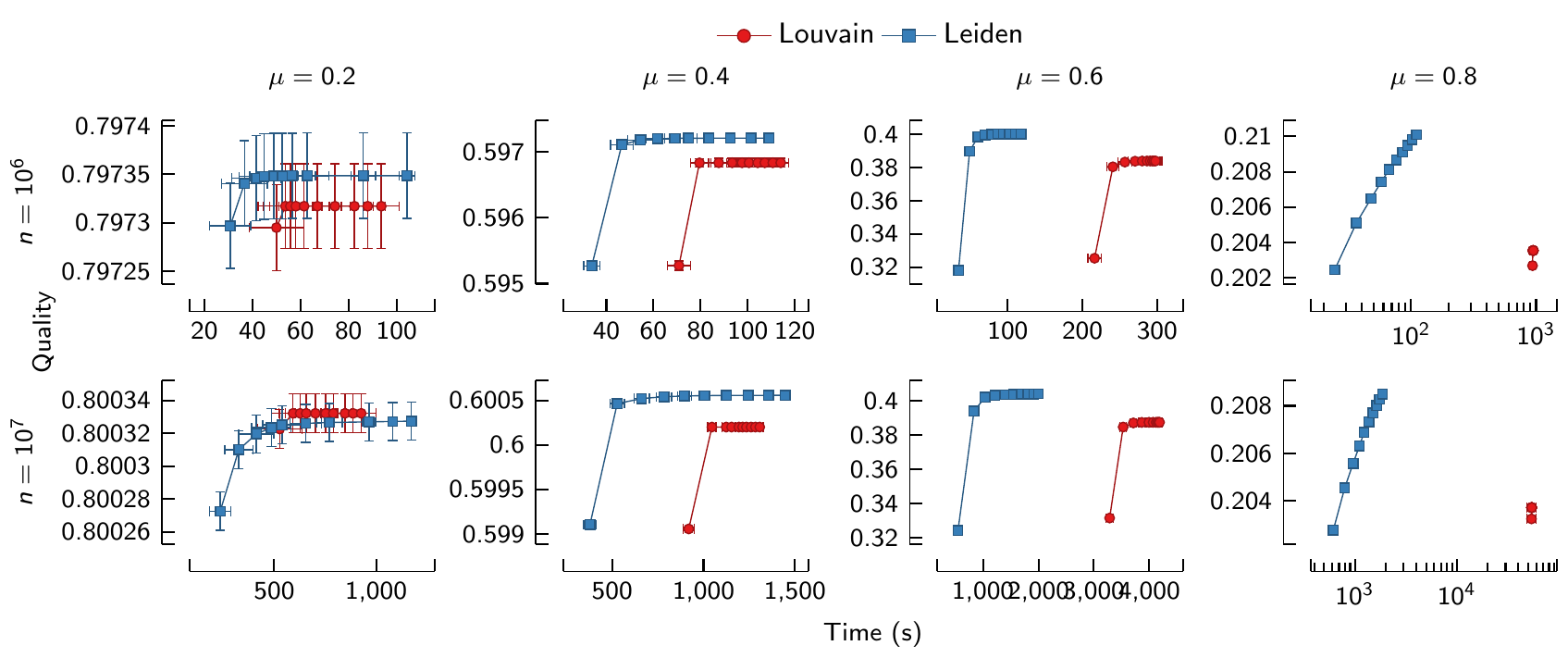}
  \end{center}
  \caption{
    \textbf{Runtime versus quality for benchmark networks}.
    Speed and quality for the first $10$ iterations of the Louvain and the Leiden algorithm for benchmark networks ($n=10^6$ and $n=10^7$).
    The horizontal axis indicates the cumulative time taken to obtain the quality indicated on the vertical axis.
    Each point corresponds to a certain iteration of an algorithm, with results averaged over $10$ experiments.
    In general, Leiden is both faster than Louvain and finds better partitions.
  }
  \label{fig:benchmark_all_itr}
\end{figure*}

\begin{figure}[tb]
  \begin{center}
    \includegraphics{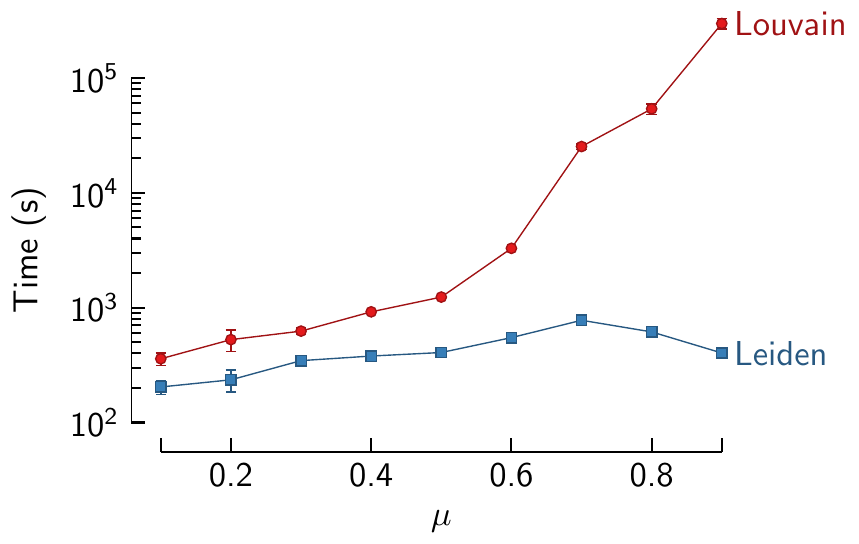}
  \end{center}
  \caption{
    \textbf{Scaling of benchmark results for difficulty of the partition}.
    Speed of the first iteration of the Louvain and the Leiden algorithm for benchmark networks with increasingly difficult partitions ($n=10^7$).
    In the most difficult case ($\mu=0.9$), Louvain requires almost $2.5$ days, while Leiden needs fewer than $10$ minutes.
  }
  \label{fig:benchmark_mixing}
\end{figure}

\subsection{Benchmark networks}

To study the scaling of the Louvain and the Leiden algorithm, we rely on a variant of a well-known approach for constructing benchmark networks~\cite{Lancichinetti2008a}.
We generated benchmark networks in the following way.
First, we created a specified number of nodes and we assigned each node to a community.
Communities were all of equal size.
A community size of $50$ nodes was used for the results presented below, but larger community sizes yielded qualitatively similar results.
We then created a certain number of edges such that a specified average degree $\langle k \rangle$ was obtained.
For the results reported below, the average degree was set to $\langle k \rangle = 10$.
Edges were created in such a way that an edge fell between two communities with a probability $\mu$ and within a community with a probability $1 - \mu$.
We applied the Louvain and the Leiden algorithm to exactly the same networks, using the same seed for the random number generator.
For both algorithms, $10$ iterations were performed.
We used the CPM quality function.
The value of the resolution parameter was determined based on the so-called mixing parameter $\mu$~\cite{Traag2011}.
We generated networks with $n=10^3$ to $n=10^7$ nodes.
For each set of parameters, we repeated the experiment $10$ times.
Below, the quality of a partition is reported as $\frac{\Hf}{2m}$, where $\Hf$ is defined in Eq.~(\ref{eq:CPM_simple}) and $m$ is the number of edges.

As shown in Fig.~\ref{fig:benchmark_nodes}, for lower values of $\mu$ the partition is well defined, and neither the Louvain nor the Leiden algorithm has a problem in determining the correct partition in only two iterations.
Hence, for lower values of $\mu$, the difference in quality is negligible.
However, as $\mu$ increases, the Leiden algorithm starts to outperform the Louvain algorithm.
The differences are not very large, which is probably because both algorithms find partitions for which the quality is close to optimal, related to the issue of the degeneracy of quality functions~\cite{Good2010}.

The Leiden algorithm is clearly faster than the Louvain algorithm.
For lower values of $\mu$, the correct partition is easy to find and Leiden is only about twice as fast as Louvain.
However, for higher values of $\mu$, Leiden becomes orders of magnitude faster than Louvain, reaching $10$--$100$ times faster runtimes for the largest networks.
As can be seen in Fig.~\ref{fig:benchmark_mixing}, whereas Louvain becomes much slower for more difficult partitions, Leiden is much less affected by the difficulty of the partition.

Fig.~\ref{fig:benchmark_all_itr} presents total runtime versus quality for all iterations of the Louvain and the Leiden algorithm.
As can be seen in the figure, Louvain quickly reaches a state in which it is unable to find better partitions.
On the other hand, Leiden keeps finding better partitions, especially for higher values of $\mu$, for which it is more difficult to identify good partitions.
A number of iterations of the Leiden algorithm can be performed before the Louvain algorithm has finished its first iteration.
Later iterations of the Louvain algorithm are very fast, but this is only because the partition remains the same.
With one exception ($\mu=0.2$ and $n=10^7$), all results in Fig.~\ref{fig:benchmark_all_itr} show that Leiden outperforms Louvain in terms of both computational time and quality of the partitions.

\subsection{Empirical networks}

Analyses based on benchmark networks have only a limited value because these networks are not representative of empirical real-world networks.
In particular, benchmark networks have a rather simple structure.
Empirical networks show a much richer and more complex structure.
We now compare how the Leiden and the Louvain algorithm perform for the six empirical networks listed in Table~\ref{tbl:real_networks}.
Our analysis is based on modularity with resolution parameter $\gamma = 1$.
For each network, Table~\ref{tbl:real_networks} reports the maximal modularity obtained using the Louvain and the Leiden algorithm.

\begin{figure}[tb]
  \begin{center}
    \includegraphics{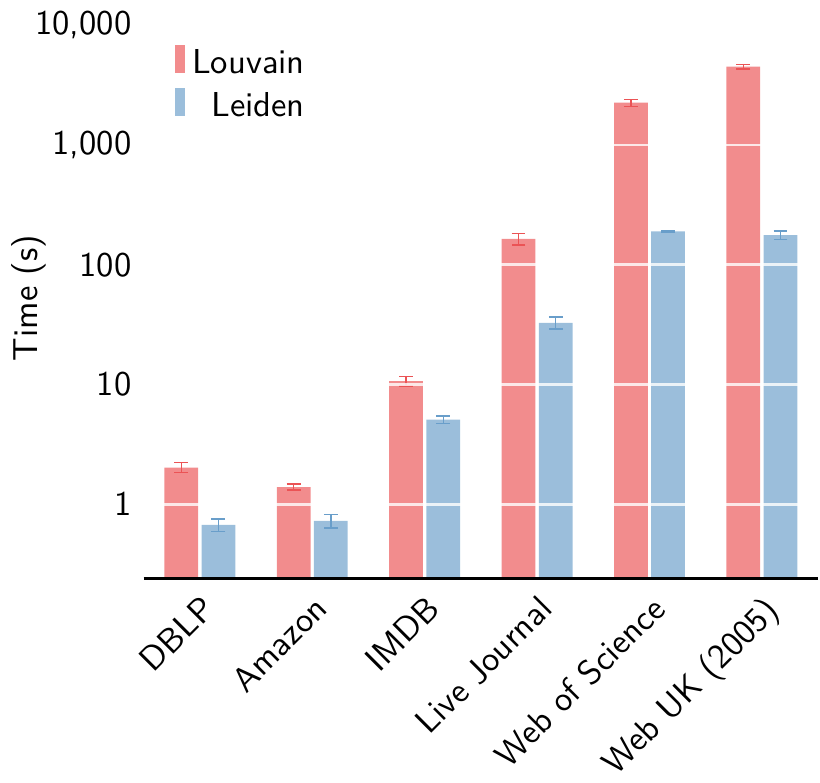}
  \end{center}
  \caption{
    \textbf{First iteration runtime for empirical networks}.
    Speed of the first iteration of the Louvain and the Leiden algorithm for six empirical networks.
    Leiden is faster than Louvain especially for larger networks.
  }
  \label{fig:real_network_1st_itr_speed}
\end{figure}

As can be seen in Fig.~\ref{fig:real_network_1st_itr_speed}, the Leiden algorithm is significantly faster than the Louvain algorithm also in empirical networks.
In the first iteration, Leiden is roughly $2$--$20$ times faster than Louvain.
The speed difference is especially large for larger networks.
This is similar to what we have seen for benchmark networks.
For the Amazon and IMDB networks, the first iteration of the Leiden algorithm is only about $1.6$ times faster than the first iteration of the Louvain algorithm.
However, Leiden is more than $7$ times faster for the Live Journal network, more than $11$ times faster for the Web of Science network and more than $20$ times faster for the Web UK network.
In fact, for the Web of Science and Web UK networks, Fig.~\ref{fig:real_networks_all_itr} shows that more than $10$ iterations of the Leiden algorithm can be performed before the Louvain algorithm has finished its first iteration.

As shown in Fig.~\ref{fig:real_networks_all_itr}, the Leiden algorithm also performs better than the Louvain algorithm in terms of the quality of the partitions that are obtained.
For all networks, Leiden identifies substantially better partitions than Louvain.
Louvain quickly converges to a partition and is then unable to make further improvements.
In contrast, Leiden keeps finding better partitions in each iteration.

The quality improvement realised by the Leiden algorithm relative to the Louvain algorithm is larger for empirical networks than for benchmark networks.
Hence, the complex structure of empirical networks creates an even stronger need for the use of the Leiden algorithm.
Leiden keeps finding better partitions for empirical networks also after the first $10$ iterations of the algorithm.
This contrasts to benchmark networks, for which Leiden often converges after a few iterations.
For empirical networks, it may take quite some time before the Leiden algorithm reaches its first stable iteration.
As can be seen in Fig.~\ref{fig:real_networks_n_itr}, for the IMDB and Amazon networks, Leiden reaches a stable iteration relatively quickly, presumably because these networks have a fairly simple community structure.
The DBLP network is somewhat more challenging, requiring almost $80$ iterations on average to reach a stable iteration.
The Web of Science network is the most difficult one.
For this network, Leiden requires over $750$ iterations on average to reach a stable iteration.
Importantly, the first iteration of the Leiden algorithm is the most computationally intensive one, and subsequent iterations are faster.
For example, for the Web of Science network, the first iteration takes about $110$--$120$ seconds, while subsequent iterations require about $40$ seconds.

\begin{figure}[tb]
  \begin{center}
    \includegraphics{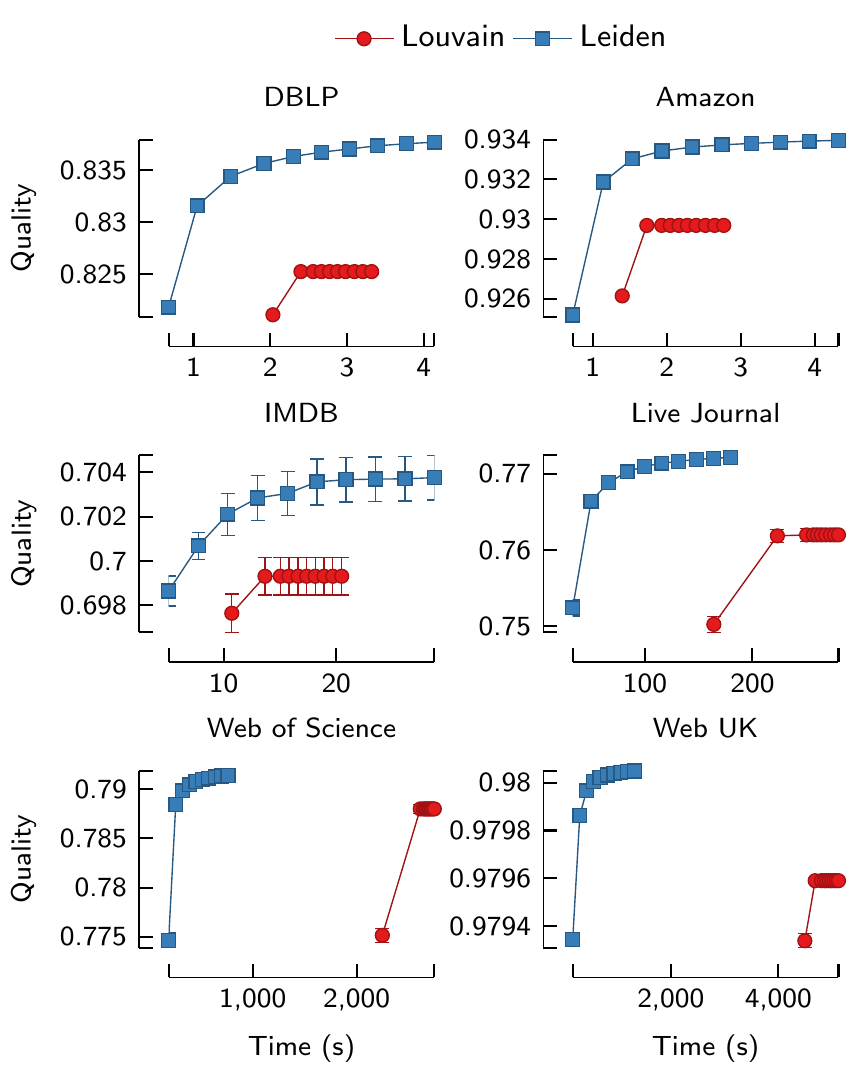}
  \end{center}
  \caption{
    \textbf{Runtime versus quality for empirical networks}.
    Speed and quality for the first $10$ iterations of the Louvain and the Leiden algorithm for six empirical networks.
    The horizontal axis indicates the cumulative time taken to obtain the quality indicated on the vertical axis.
    Each point corresponds to a certain iteration of an algorithm, with results averaged over $10$ experiments.
    Leiden is both faster than Louvain and finds better partitions.
  }
  \label{fig:real_networks_all_itr}
\end{figure}

\begin{figure}[tb]
  \begin{center}
    \includegraphics[width=\columnwidth]{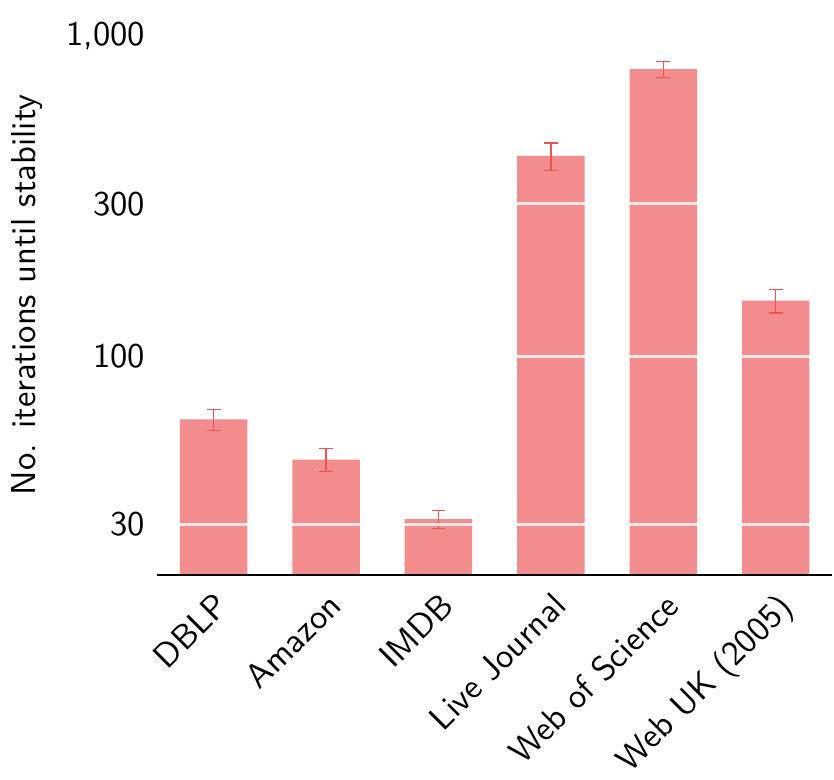}%
  \end{center}
  \caption{
    \textbf{Number of iterations until stability}.
    Number of iterations before the Leiden algorithm has reached a stable iteration for six empirical networks.
    In a stable iteration, the partition is guaranteed to be node optimal and subpartition $\gamma$-dense.
  }
  \label{fig:real_networks_n_itr}
\end{figure}

\section{Discussion}

Community detection is an important task in the analysis of complex networks.
Finding communities in large networks is far from trivial: algorithms need to be fast, but they also need to provide high-quality results.
One of the most widely used algorithms is the Louvain algorithm~\cite{Blondel2008}, which is reported to be among the fastest and best performing community detection algorithms~\cite{Lancichinetti2009,Yang2016}.
However, as shown in this paper, the Louvain algorithm has a major shortcoming: the algorithm yields communities that may be arbitrarily badly connected.
Communities may even be disconnected.

To overcome the problem of arbitrarily badly connected communities, we introduced a new algorithm, which we refer to as the Leiden algorithm.
This algorithm provides a number of explicit guarantees.
In particular, it yields communities that are guaranteed to be connected.
Moreover, when the algorithm is applied iteratively, it converges to a partition in which all subsets of all communities are guaranteed to be locally optimally assigned.
In practical applications, the Leiden algorithm convincingly outperforms the Louvain algorithm, both in terms of speed and in terms of quality of the results, as shown by the experimental analysis presented in this paper.
We conclude that the Leiden algorithm is strongly preferable to the Louvain algorithm.

\bibliography{bibliography}

\begin{acknowledgments}
\noindent
We gratefully acknowledge computational facilities provided by the LIACS Data Science Lab Computing Facilities through Frank Takes. We thank Lovro \u{S}ubelj for his comments on an earlier version of this paper.
\end{acknowledgments}

\section*{Author contributions statement}
\noindent
All authors conceived the algorithm and contributed to the source code.
VAT performed the experimental analysis.
VAT and LW wrote the manuscript.
NJvE reviewed the manuscript.
\section*{Additional information}

\noindent
\textbf{Competing interests}

\noindent
The authors act as bibliometric consultants to CWTS B.V., which makes use of community detection algorithms in commercial products and services.

\clearpage

\onecolumngrid
\appendix
\counterwithin{figure}{section}
\counterwithin{algorithm}{section}

\section{Pseudo-code and mathematical notation}
\label{sec:code_notation}

Pseudo-code for the Louvain algorithm and the Leiden algorithm is provided in Algorithms~\ref{algo:louvain} and~\ref{algo:leiden}, respectively.
Below we discuss the mathematical notation that is used in the pseudo-code and also in the mathematical results presented in Appendices~\ref{sec:reachability}, \ref{sec:guarantees}, and~\ref{sec:bound_on_optimality}.
There are some uncommon elements in the notation.
In particular, the idea of sets of sets plays an important role, and some concepts related to this idea need to be introduced.

Let $G = (V, E)$ be a graph with $n = |V|$ nodes and $m = |E|$ edges.
Graphs are assumed to be undirected.
With the exception of Theorem~\ref{thm:bound_external_edges} in Appendix~\ref{sec:bound_on_optimality}, the mathematical results presented in this paper apply to both unweighted and weighted graphs.
For simplicity, our mathematical notation assumes graphs to be unweighted, although the notation does allow for multigraphs.
A partition $\P = \{C_1, \ldots, C_r\}$ consists of $r = |\P|$ communities, where each community $C_i \subseteq V$ consists of a set of nodes such that $V = \bigcup_i C_i$ and $C_i \cap C_j = \emptyset$ for all $i \neq j$.
For two sets $R$ and $S$, we sometimes use $R + S$ to denote the union $R \cup S$ and $R - S$ to denote the difference $R \setminus S$.

A quality function $\Hf(G, \P)$ assigns a ``quality'' to a partition $\P$ of a graph $G$.
We aim to find a partition with the highest possible quality.
The graph $G$ is often clear from the context, and we therefore usually write $\Hf(\P)$ instead of $\Hf(G, \P)$.
Based on partition $\P$, graph $G$ can be \emph{aggregated} into a new graph $G'$.
Graph $G$ is then called the \emph{base graph}, while graph $G'$ is called the \emph{aggregate graph}.
The nodes of the aggregate graph $G'$ are the communities in the partition $\P$ of the base graph $G$, i.e. $V(G') = \P$.
The edges of the aggregate graph $G'$ are multi-edges.
The number of edges between two nodes in the aggregate graph $G'$ equals the number of edges between nodes in the two corresponding communities in the base graph $G$.
Hence, $E(G') = \{ (C, D) \mid (u, v) \in E(G), u \in C \in \P, v \in D \in \P \}$, where $E(G')$ is a multiset.
A quality function must have the property that $\Hf(G, \P) = \Hf(G', \P')$, where $\P' = \{\{v\} \mid v \in V(G')\}$ denotes the singleton partition of the aggregate graph $G'$.
This ensures that a quality function gives consistent results for base graphs and aggregate graphs.

We denote by $\P(v \mapsto C)$ the partition that is obtained when we start from partition $\P$ and we then move node $v$ to community $C$.
We write $\Delta \Hf_\P(v \mapsto C)$ for the change in the quality function by moving node $v$ to community $C$ for some partition $\P$.
In other words, $\Delta \Hf_\P(v \mapsto C) = \Hf(\P(v \mapsto C)) - \Hf(\P)$.
We usually leave the partition $\P$ implicit and simply write $\Delta \Hf(v \mapsto C)$.
Similarly, we denote by $\Delta \Hf_\P(S \mapsto C)$ the change in the quality function by moving a set of nodes $S$ to community $C$.
An empty community is denoted by $\emptyset$.
Hence, $\Delta \Hf_\P(S \mapsto \emptyset)$ is the change in the quality function by moving a set of nodes $S$ to an empty (i.e. new) community.

Now consider a community $C$ that consists of two parts $S_1$ and $S_2$ such that $C = S_1 \cup S_2$ and $S_1 \cap S_2 = \emptyset$.
Suppose that $S_1$ and $S_2$ are disconnected.
In other words, there are no edges between nodes in $S_1$ and $S_2$.
We then require a quality function to have the property that $\Delta \Hf(S_1 \mapsto \emptyset) > 0$ and $\Delta \Hf(S_2 \mapsto \emptyset) > 0$.
This guarantees that a partition can always be improved by splitting a community into its connected components.
This comes naturally for most definitions of a community, but this is not the case when considering for example negative links~\cite{Traag2009}.

Because nodes in an aggregate graph are sets themselves, it is convenient to define some recursive properties.
\begin{definition}
  The \emph{recursive size} of a set $S$ is defined as
  \begin{equation}
    \|S\| = \sum_{s \in S} \|s\|,
    \label{eq:set_size}
  \end{equation}
  where $\|s\| = 1$ if $s$ is not a set itself.
  The \emph{flattening} operation for a set $S$ is defined as
  \begin{equation}
    \flatf(S) = \bigcup_{s \in S} \flatf(s),
  \end{equation}
  where $\flatf(s) = s$ if $s$ is not a set itself.
  A set that has been flattened is called a \emph{flat set}.
\end{definition}
The recursive size of a set corresponds to the usual definition of set size in case the elements of a set are not sets themselves, but it generalizes this definition whenever the elements are sets themselves.
For example, if $S = \{ \{a, b\}, \{c\}, \{d, e, f\} \}$, then
\begin{align*}
  \|S\| &= \|\{a, b\}\| + \|\{c\}\| + \|\{d, e, f\}\| \\
    &= (\|a\| + \|b\|) + \|c\| + (\|d\| + \|e\| + \|f\|) \\
    &= 2 + 1 + 3 = 6.
\end{align*}
This contrasts with the traditional size of a set, which is $|S| = 3$, because $S$ contains $3$ elements.
The fact that the elements are sets themselves plays no role in the traditional size of a set.
The flattening of $S$ is
\begin{align*}
  \flatf(S) &= \flatf(\{a, b\}) \cup \flatf(\{c\}) \cup \flatf(\{d, e, f\}) \\
    &= a \cup b \cup c \cup d \cup e \cup f \\
    &= \{a, b, c, d, e, f\}.
\end{align*}
Note that $\|S\| = |\flatf(S)|$.
\begin{definition}
  The \emph{flattening} operation for a partition $\P$ is defined as
  \begin{equation}
    \flatf^*(\P) = \{ \flatf(C) \mid C \in \P \}.
  \end{equation}
  Hence, $\flatf^*(\P)$ denotes the operation in which each community $C \in \P$ is flattened.
  A partition that has been flattened is called a \emph{flat partition}.
\end{definition}
For any partition of an aggregate graph, the equivalent partition of the base graph can be obtained by applying the flattening operation.

Additionally, we need some terminology to describe the connectivity of communities.
\begin{definition}
  Let $G = (V, E)$ be a graph, and let $\P$ be a partition of $G$.
  Furthermore, let $H(C)$ be the subgraph induced by a community $C \in \P$, i.e. $V(H) = C$ and $E(H) = \{(u, v) \mid (u, v) \in E(G), u, v \in C \}$.
  A community $C \in \P$ is called \emph{connected} if $H(C)$ is a connected graph.
  Conversely, a community $C \in \P$ is called \emph{disconnected} if $H(C)$ is a disconnected graph.
\end{definition}

The mathematical proofs presented in this paper rely on the Constant Potts Model (CPM)~\cite{Traag2011}.
This quality function has important advantages over modularity.
In particular, unlike modularity, CPM does not suffer from the problem of the resolution limit~\cite{Fortunato:2007p183,Traag2011}.
Moreover, our mathematical definitions and proofs are quite elegant when expressed in terms of CPM.
The CPM quality function is defined as
\begin{equation}
  \Hf(G, \P) = \sum_{C \in \P} \left[E(C, C) - \gamma \binom{\|C\|}{2}\right],
  \label{eq:CPM}
\end{equation}
where $E(C, D) = |\{(u, v) \in E(G) \mid u \in C, v \in D\}|$ denotes the number of edges between nodes in communities $C$ and $D$.
Note that this definition can also be used for aggregate graphs because $E(G)$ is a multiset.

The mathematical results presented in this paper also extend to modularity, although the formulations are less elegant.
Results for modularity are straightforward to prove by redefining the recursive size $\|S\|$ of a set $S$.
We need to define the size of a node $v$ in the base graph as $\|v\| = k_v$ instead of $\|v\| = 1$, where $k_v$ is the degree of node $v$.
Furthermore, we need to rescale the resolution parameter $\gamma$ by $2m$.
Modularity can then be written as
\begin{equation}
  \Hf(G, \P) = \sum_{C \in \P} \left[E(C, C) - \frac{\gamma}{2m} \binom{\|C\|}{2}\right].
  \label{eq:modularity}
\end{equation}
Note that, in addition to the overall multiplicative factor of $\frac{1}{2m}$, this adds a constant $\frac{\gamma}{2m} \sum_C \frac{\|C\|}{2} = \frac{\gamma}{2}$ to the ordinary definition of modularity~\cite{Newman2004Finding}.
However, this does not matter for optimisation or for the proofs.

As discussed in the main text, the Louvain and the Leiden algorithm can be \emph{iterated} by performing multiple consecutive iterations of the algorithm, using the partition identified in one iteration as starting point for the next iteration.
In this way, a sequence of partitions $\P_0, \P_1, \ldots$ is obtained such that $\P_{t + 1} = \textsc{Louvain}(G, \P_t)$ or $\P_{t + 1} = \textsc{Leiden}(G, \P_t)$.
The initial partition $\P_0$ usually is the singleton partition of the graph $G$, i.e. $\P_0 = \{\{v\} \mid v \in V\}$.

\begin{algorithm*}[bt]
  \begin{algorithmic}[1]
    \Function{Louvain}{Graph $G$, Partition $\P$}
      \Do
        \State $\P \gets \Call{MoveNodes}{G, \P}$ \Comment{Move nodes between communities}
        \State done $ \gets |\P| = |V(G)|$ \Comment{Terminate when each community consists of only one node}
        \If{not done}
          \State $G \gets \Call{AggregateGraph}{G, \P}$ \Comment{Create aggregate graph based on partition $\P$}
          \State $\P \gets \Call{SingletonPartition}{G}$ \Comment{Assign each node in aggregate graph to its own community}
        \EndIf
      \doWhile{not done}
      \State \Return $\flatf^*(\P)$
    \EndFunction
    \item[]
    \Function{MoveNodes}{Graph $G$, Partition $\P$}
      \Do
        \State $\Hf_\text{old} = \Hf(\P)$
        \For{$v \in V(G)$} \Comment{Visit nodes (in random order)}
          \State $C' \gets \argmax_{C \in \P \cup \emptyset} \Delta\Hf_\P(v \mapsto C)$ \Comment{Determine best community for node $v$}
          \If{$\Delta\Hf_\P(v \mapsto C') > 0$} \Comment{Perform only strictly positive node movements} \label{algo:louvain:strict_increase_move}
            \State $v \mapsto C'$ \Comment{Move node $v$ to community $C'$}
          \EndIf
        \EndFor
      \doWhile{$\Hf(\P) > \Hf_\text{old}$} \Comment{Continue until no more nodes can be moved}
      \State \Return $\P$
    \EndFunction
    \item[]
    \Function{AggregateGraph}{Graph $G$, Partition $\P$}
      \State $V \gets \P$ \Comment{Communities become nodes in aggregate graph}
      \State $E \gets \{(C, D) \mid (u, v) \in E(G), u \in C \in \P, v \in D \in \P\}$ \Comment{$E$ is a multiset}
      \State \Return $\Call{Graph}{V, E}$
    \EndFunction
    \item[]
    \Function{SingletonPartition}{Graph $G$}
      \State \Return $\{ \{v\} \mid v \in V(G) \}$ \Comment{Assign each node to its own community}
    \EndFunction
  \end{algorithmic}
  \caption{\textbf{Louvain algorithm.}}
  \label{algo:louvain}
\end{algorithm*}

\begin{algorithm*}[bt]
  \begin{algorithmic}[1]
    \Function{Leiden}{Graph $G$, Partition $\P$}
      \Do
        \State $\P \gets \Call{MoveNodesFast}{G, \P}$ \Comment{Move nodes between communities} \label{algo:leiden:call_move_nodes}
        \State done $ \gets |\P| = |V(G)|$ \Comment{Terminate when each community consists of only one node} \label{algo:leiden:stop_criterion}
        \If{not done}
          \State $\P_\text{refined} \gets \Call{RefinePartition}{G, \P}$ \Comment{Refine partition $\P$}
          \State $G \gets \Call{AggregateGraph}{G, \P_\text{refined}}$ \Comment{Create aggregate graph based on refined partition $\P_\text{refined}$}
          \State $\P \gets \{\{v \mid v \subseteq C, v \in V(G)\} \mid C \in \P\}$ \Comment{But maintain partition $\P$}
        \EndIf
      \doWhile{not done}
      \State \Return $\flatf^*(\P)$
    \EndFunction
    \item[]
    \Function{MoveNodesFast}{Graph $G$, Partition $\P$}
      \State $Q \gets \Call{Queue}{V(G)}$ \Comment{Make sure that all nodes will be visited (in random order)}
	  \Do
        \State $v \gets Q$.remove() \Comment{Determine next node to visit}
        \State $C' \gets \argmax_{C \in \P \cup \emptyset} \Delta\Hf_\P(v \mapsto C)$ \Comment{Determine best community for node $v$}
        \If{$\Delta\Hf_\P(v \mapsto C') > 0$} \Comment{Perform only strictly positive node movements} \label{algo:leiden:strict_increase_move}
          \State $v \mapsto C'$ \Comment{Move node $v$ to community $C'$}
          \State $N \gets \{u \mid (u, v) \in E(G), u \notin C'\}$ \Comment{Identify neighbours of node $v$ that are not in community $C'$}
          \State $Q$.add($N - Q$) \Comment{Make sure that these neighbours will be visited} \label{algo:leiden:add_neighs_to_queue}
        \EndIf
      \doWhile{$Q \neq \emptyset$} \Comment{Continue until there are no more nodes to visit}
      \State \Return $\P$
    \EndFunction
    \item[]
    \Function{RefinePartition}{Graph $G$, Partition $\P$}
      \State $\P_\text{refined} \gets \Call{SingletonPartition}{G}$ \Comment{Assign each node to its own community}
      \For{$C \in \P$} \Comment{Visit communities}
        \State $\P_\text{refined} \gets \Call{MergeNodesSubset}{G, \P_\text{refined}, C}$ \Comment{Refine community $C$}
      \EndFor
      \State \Return $\P_\text{refined}$
    \EndFunction
    \item[]
    \Function{MergeNodesSubset}{Graph $G$, Partition $\P$, Subset $S$}
	  \State $R = \{v \mid v \in S, E(v, S - v) \geq \gamma \|v\| \cdot (\|S\| - \|v\|)\}$ \Comment{Consider only nodes that are well connected within subset $S$} \label{algo:leiden:strict_merge1}
      \For{$v \in R$} \Comment{Visit nodes (in random order)}
        \If{$v$ in singleton community} \Comment{Consider only nodes that have not yet been merged}
          \State $\T \gets \{C \mid C \in \P, C \subseteq S, E(C, S - C) \geq \gamma \|C\| \cdot (\|S\| - \|C\|)\}$ \Comment{Consider only well-connected communities} \label{algo:leiden:strict_merge2}
          \State $\Pr(C' = C) \sim
            \begin{cases}
              \exp\left( \frac{1}{\theta} \Delta \Hf_\P(v \mapsto C) \right) & \text{if~} \Delta \Hf_\P(v \mapsto C) \geq 0 \\
              0 & \text{otherwise}
            \end{cases}$ \quad for $C \in \T$ \Comment{Choose random community $C'$} \label{algo:leiden:merge_prob}
          \State $v \mapsto C'$ \Comment{Move node $v$ to community $C'$}
        \EndIf
      \EndFor
      \State \Return $\P$
    \EndFunction
    \item[]
    \Function{AggregateGraph}{Graph $G$, Partition $\P$}
      \State $V \gets \P$ \Comment{Communities become nodes in aggregate graph}
      \State $E \gets \{(C, D) \mid (u, v) \in E(G), u \in C \in \P, v \in D \in \P\}$ \Comment{$E$ is a multiset}
      \State \Return $\Call{Graph}{V, E}$
    \EndFunction
    \item[]
    \Function{SingletonPartition}{Graph $G$}
      \State \Return $\{ \{v\} \mid v \in V(G) \}$ \Comment{Assign each node to its own community}
    \EndFunction
  \end{algorithmic}
  \caption{\textbf{Leiden algorithm.}}
  \label{algo:leiden}
\end{algorithm*}

\section{Disconnected communities in the Louvain algorithm}
\label{sec:disconnected_example}

In this appendix, we analyse the problem that communities obtained using the Louvain algorithm may be disconnected.
This problem is also discussed in the main text (Section~\ref{sec:disconnected}), using the example presented in Fig.~\ref{fig:disconnected_community}.
However, the main text offers no numerical details.
These details are provided below.

We consider the CPM quality function with a resolution of $\gamma = \frac{1}{7}$.
In the example presented in Fig.~\ref{fig:disconnected_community}, the edges between nodes 0 and 1 and between nodes 0 and 4 have a weight of $2$, as indicated by the thick lines in the figure.
All other edges have a weight of $1$.
The Louvain algorithm starts from a singleton partition, with each node being assigned to its own community.
The algorithm then keeps iterating over all nodes, moving each node to its optimal community.
Depending on the order in which the nodes are visited, the following could happen.
Node 1 is visited first, followed by node 4.
Nodes 1 and 4 join the community of node 0, because the weight of the edges between nodes 0 and 1 and between nodes 0 and 4 is sufficiently high.
For node 1, the best move clearly is to join the community of node 0.
For node 4, the benefit of joining the community of nodes 0 and 1 then is $2 - \gamma \cdot 2 = \frac{12}{7}$.
This is larger than the benefit of joining the community of node 5 or 6, which is $1 - \gamma \cdot 1 = \frac{6}{7}$.
Next, nodes 2, 3, 5 and 6 are visited.
For these nodes, it is beneficial to join the community of nodes 0, 1 and 4, because joining this community has a benefit of at least $1 - \gamma \cdot 6 = \frac{1}{7} > 0$.
This then yields the situation portrayed in Fig.~\ref{fig:disconnected_community}(a).
After some node movements in the rest of the graph, some neighbours of node 0 in the rest of the graph end up together in a new community.
Consequently, when node 0 is visited, it can best be moved to this new community, which gives the situation depicted in Fig.~\ref{fig:disconnected_community}(b).
In particular, suppose there are $5$ nodes in the new community, all of which are connected to node 0.
In that case, the benefit for node 0 of moving to this community is $5 - \gamma \cdot 5 = \frac{30}{7}$, while the benefit of staying in the current community is only $2 \cdot 2 - \gamma \cdot 6 = \frac{22}{7}$.
After node 0 has moved, nodes 1 and 4 are still locally optimally assigned.
For these nodes, the benefit of moving to the new community of node 0 is $2 - \gamma \cdot 6 = \frac{8}{7}$.
This is smaller than the benefit of staying in the current community, which is $2 - \gamma \cdot 5 = \frac{9}{7}$.
Finally, nodes 2, 3, 5 and 6 are all locally optimally assigned, as $1 - \gamma \cdot 5 = \frac{2}{7} > 0$.
Hence, we end up with a community that is disconnected.
In later stages of the Louvain algorithm, there will be no possibility to repair this.

The example presented above considers a weighted graph, but this graph can be assumed to be an aggregate graph of an unweighted base graph, thus extending the example also to unweighted graphs.
Although the example uses the CPM quality function, similar examples can be given for modularity.
However, because of the dependency of modularity on the number of edges $m$, the calculations for modularity are a bit more complex.
Importantly, both for CPM and for modularity, the Louvain algorithm suffers from the problem of disconnected communities.

\section{Reachability of optimal partitions}
\label{sec:reachability}

In this appendix, we consider two types of move sequences: non-decreasing move sequences and greedy move sequences.
For each type of move sequence, we study whether all optimal partitions are reachable.
We first show that this is not the case for greedy move sequences.
In particular, we show that for some optimal partitions there does not exist a greedy move sequence that is able to reach the partition.
We then show that optimal partitions can always be reached using a non-decreasing move sequence.
This result forms the basis for the asymptotic guarantees of the Leiden algorithm, which are discussed in Appendix~\ref{sec:asymptotic}.

We first define the different types of move sequences.
\begin{definition}
  Let $G = (V, E)$ be a graph, and let $\P_0, \ldots, \P_\tau$ be partitions of $G$.
  A sequence of partitions $\P_0, \ldots, \P_\tau$ is called a \emph{move sequence} if for each $t = 0, \ldots, \tau - 1$ there exists a node $v_t \in V$ and a community $C_t \in \P_t \cup \emptyset$ such that $\P_{t + 1} = \P_t(v_t \mapsto C_t)$.
  A move sequence is called \emph{non-decreasing} if $\Hf(\P_{t + 1}) \geq \Hf(\P_t)$ for all $t = 0, \ldots, \tau - 1$.
  A move sequence is called \emph{greedy} if $\Hf(\P_{t + 1}) = \max_C \Hf(\P_t(v_t \mapsto C))$ for all $t = 0, \ldots, \tau - 1$.
\end{definition}
In other words, the next partition in a move sequence is obtained by moving a single node to a different community.
Clearly, a greedy move sequence must be non-decreasing, but a non-decreasing move sequence does not need to be greedy.
A natural question is whether for any optimal partition $\P^*$ there exists a move sequence that starts from the singleton partition and that reaches the optimal partition, i.e., a move sequence $\P_0, \ldots, \P_\tau$ with $\P_0 = \{\{v\} \mid v \in V\}$ and $\P_\tau = \P^*$.
Trivially, it is always possible to reach the optimal partition if we allow all moves---even moves that decrease the quality function---as is done for example in simulated annealing~\cite{Guimera2005Functional,Reichardt2006Statistical}.
However, it can be shown that there is no need to consider all moves in order to reach the optimal partition.
It is sufficient to consider only non-decreasing moves.
On the other hand, considering only greedy moves turns out to be too restrictive to guarantee that the optimal partition can be reached.

\subsection{Non-decreasing move sequences}
\label{sec:nondecreasing_move_sequences}

We here prove that for any graph there exists a non-decreasing move sequence that reaches the optimal partition $\P^*$. The optimal partition can be reached in $n - |\P^*|$ steps.
\begin{theorem}
  Let $G = (V, E)$ be a graph, and let $\P^*$ be an optimal partition of $G$.
  There then exists a non-decreasing move sequence $\P_0, \ldots, \P_\tau$ with $\P_0 = \{\{v\} \mid v \in V\}$, $\P_\tau = \P^*$, and $\tau = n - |\P^*|$.
  \label{thm:optimal_seq}
\end{theorem}
\begin{proof}
  Let $C^* \in \P^*$ be a community in the optimal partition $\P^*$, let $v_0 \in C^*$ be a node in this community, and let $C_0 = \{v_0\}$.
  Let $\P_0 = \{\{v\} \mid v \in V\}$ be the singleton partition.
  For $t = 1, \ldots, |C^*| - 1$, let $v_t \in C^* - C_{t - 1}$, let $C_t = \{v_0, \ldots, v_t\} \in \P_t$, and let $\P_t = \P_{t - 1}(v_t \mapsto C_{t - 1})$.
  We prove by contradiction that there always exists a non-decreasing move sequence $\P_0, \ldots, \P_{|C^*| - 1}$.
  Assume that for some $t$ there does not exist a node $v_t$ for which $\Delta \Hf(v_t \mapsto C_{t - 1}) \geq 0$.
  Let $S = C^* - C_{t - 1}$ and $R = C_{t - 1}$.
  For all $v \in S$,
  \begin{equation*}
    E(v, R) - \gamma \|v\| \cdot \|R\| < 0.
  \end{equation*}
  This implies that
  \begin{equation*}
    E(S, R) = \sum_{v \in S} E(v, R) < \gamma \|S\| \cdot \|R\|.
  \end{equation*}
  However, by optimality, for all $S \subseteq C^*$ and $R = C^* - S$,
  \begin{equation*}
    E(S, R) \geq \gamma \|S\| \cdot \|R\|.
  \end{equation*}
  We therefore have a contradiction.
  Hence, there always exists a non-decreasing move sequence $\P_0, \ldots, \P_{|C^*| - 1}$.
  This move sequence reaches the community $C_t = C^*$.
  The above reasoning can be applied to each community $C^* \in \P^*$.
  Consequently, each of these communities can be reached using a non-decreasing move sequence.
  In addition, for each community $C^* \in \P^*$, this can be done in $|C^*| - 1$ steps, so that in total $\tau = \sum_{C^* \in \P^*} (|C^*| - 1) = n - |\P^*|$ steps are needed.
\end{proof}

\subsection{Greedy move sequences}
\label{sec:greedy_move_sequences}

\begin{figure}[bt]
  \begin{center}
    \includegraphics{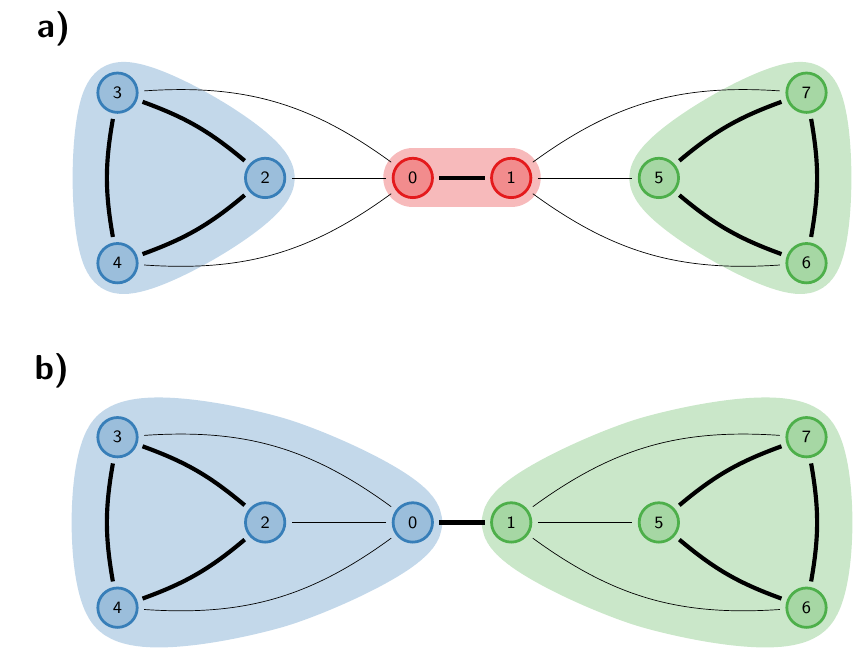}
  \end{center}
  \caption{\textbf{Unreachable optimal partition.}
    A greedy move sequence always reaches the partition in (a), whereas the partition in (b) is optimal.
    This demonstrates that for some graphs there does not exist a greedy move sequence that reaches the optimal partition.}
  \label{fig:unreachable_partition}
\end{figure}

We here show that there does not always exist a greedy move sequence that reaches the optimal partition of a graph.
To show this, we provide a counterexample in which we have a graph for which there is no greedy move sequence that reaches the optimal partition.
Our counterexample includes two nodes that should be assigned to different communities.
However, because there is a strong connection between the nodes, in a greedy move sequence the nodes are always assigned to the same community.
We use the CPM quality function in our counterexample, but a similar counterexample can be given for modularity.
The counterexample is illustrated in Fig.~\ref{fig:unreachable_partition}.
The thick edges have a weight of $3$, while the thin ones have a weight of $\frac{3}{2}$.
The resolution is set to $\gamma = 1$.
In this situation, nodes 0 and 1 are always joined together in a community.
This has a benefit of $3 - \gamma = 2$, which is larger than the benefit of $3 \cdot \frac{3}{2} - \gamma \cdot 3 = \frac{3}{2}$ obtained by node 0 joining the community of nodes 2, 3 and 4 or node 1 joining the community of nodes 5, 6 and 7.
Hence, regardless of the exact node order, the partition reached by a greedy move sequence always consists of three communities.
This gives a total quality of
\begin{equation*}
2 \cdot \left(3 \cdot 3 - \gamma \frac{3 \cdot 2}{2}\right) + \left(3 - \gamma \frac{2 \cdot 1}{2} \right) = 14,
\end{equation*}
while the optimal partition has only two communities, consisting of nodes $\{0, 2, 3, 4\}$ and $\{1, 5, 6, 7\}$ and resulting in a total quality of
\begin{equation*}
2 \cdot \left(3 \cdot 3 + 3 \cdot \frac{3}{2} - \gamma \frac{4 \cdot 3}{2} \right) = 15.
\end{equation*}
Hence, a greedy move sequence always reaches the partition in Fig.~\ref{fig:unreachable_partition}(a), whereas the partition in Fig.~\ref{fig:unreachable_partition}(b) is optimal.

\section{Guarantees of the Leiden algorithm}
\label{sec:guarantees}

In this appendix, we discuss the guarantees provided by the Leiden algorithm.
The guarantees of the Leiden algorithm partly rely on the randomness in the algorithm.
We therefore require that $\theta > 0$.
Before stating the guarantees of the Leiden algorithm, we first define a number of properties.
We start by introducing some relatively weak properties, and we then move on to stronger properties.
In the following definitions, $\P$ is a flat partition of a graph $G = (V, E)$.

\begin{definition}[$\gamma$-separation]
  We call a pair of communities $C, D \in \P$ \emph{$\gamma$-separated} if $\Delta \Hf(C \mapsto D) = \Delta \Hf(D \mapsto C) \leq 0$.
  A community $C \in \P$ is $\gamma$-separated if $C$ is $\gamma$-separated with respect to all $D \in \P$.
  A partition $\P$ is $\gamma$-separated if all $C \in \P$ are $\gamma$-separated.
  \label{def:gamma_separation}
\end{definition}

\begin{definition}[$\gamma$-connectivity]
  We call a set of nodes $S \subseteq C \in \P$ \emph{$\gamma$-connected} if $|S| = 1$ or if $S$ can be partitioned into two sets $R$ and $T$ such that $E(R, T) \geq \gamma \|R\| \cdot \|T\|$ and $R$ and $T$ are $\gamma$-connected.
  A community $C \in \P$ is $\gamma$-connected if $S = C$ is $\gamma$-connected.
  A partition $\P$ is $\gamma$-connected if all $C \in \P$ are $\gamma$-connected.
  \label{def:gamma_connectivity}
\end{definition}

\begin{definition}[Subpartition $\gamma$-density]
  We call a set of nodes $S \subseteq C \in \P$ \emph{subpartition $\gamma$-dense} if the following two conditions are satisfied: (i) $\Delta \Hf(S \mapsto \emptyset) \leq 0$ and (ii) $|S| = 1$ or $S$ can be partitioned into two sets $R$ and $T$ such that $E(R, T) \geq \gamma \|R\| \cdot \|T\|$ and $R$ and $T$ are subpartition $\gamma$-dense.
  A community $C \in \P$ is subpartition $\gamma$-dense if $S = C$ is subpartition $\gamma$-dense.
  A partition $\P$ is subpartition $\gamma$-dense if all $C \in \P$ are subpartition $\gamma$-dense.
  \label{def:subpartition_gamma_density}
\end{definition}

\begin{definition}[Node optimality]
  We call a community $C \in \P$ \emph{node optimal} if $\Delta \Hf(v \mapsto D) \leq 0$ for all $v \in C$ and all $D \in \P$ (or $D = \emptyset$).
  A partition $\P$ is node optimal if all $C \in \P$ are node optimal.
  \label{def:node_optimality}
\end{definition}

\begin{definition}[Uniform $\gamma$-density]
  We call a community $C \in \P$ \emph{uniformly $\gamma$-dense} if $\Delta \Hf(S \mapsto \emptyset) \leq 0$ for all $S \subseteq C$.
  A partition $\P$ is uniformly $\gamma$-dense if all $C \in \P$ are uniformly $\gamma$-dense.
  \label{def:uniform_gamma_density}
\end{definition}

\begin{definition}[Subset optimality]
  We call a community $C \in \P$ \emph{subset optimal} if $\Delta \Hf(S \mapsto D) \leq 0$ for all $S \subseteq C$ and all $D \in \P$ (or $D = \emptyset$).
  A partition $\P$ is subset optimal if all $C \in \P$ are subset optimal.
  \label{def:subset_optimality}
\end{definition}

Subset optimality clearly is the strongest property and subsumes all other properties.
Uniform $\gamma$-density is subsumed by subset optimality but may be somewhat more intuitive to grasp.
It states that any subset of nodes in a community is always connected to the rest of the community with a density of at least $\gamma$.
In other words, for all $S \subseteq C \in \P$ we have
\begin{equation}
  E(S, C - S) \geq \gamma \|S\| \cdot \|C - S\|.
\end{equation}
Imposing the restriction $D = \emptyset$ in the definition of subset optimality gives the property of uniform $\gamma$-density, restricting $S$ to consist of only one node gives the property of node optimality, and imposing the restriction $S = C$ yields the property of $\gamma$-separation.
Uniform $\gamma$-density implies subpartition $\gamma$-density, which in turn implies $\gamma$-connectivity.
Subpartition $\gamma$-density also implies that individual nodes cannot be split from their community (but notice that this is a weaker property than node optimality).
Ordinary connectivity is implied by $\gamma$-connectivity, but not vice versa.
Obviously, any optimal partition is subset optimal, but not the other way around: a subset optimal partition is not necessarily an optimal partition (see Fig.~\ref{fig:unreachable_partition}(a) for an example).

In the rest of this appendix, we show that the Leiden algorithm guarantees that the above properties hold for partitions produced by the algorithm.
The properties hold either in each iteration, in every stable iteration, or asymptotically.
The first two properties of $\gamma$-separation and $\gamma$-connectivity are guaranteed in each iteration of the Leiden algorithm.
We prove this in Appendix~\ref{sec:each_iteration}.
The next two properties of subpartition $\gamma$-density and node optimality are guaranteed in every stable iteration of the Leiden algorithm, as we prove in Appendix~\ref{sec:stable_iteration}.
Finally, in Appendix~\ref{sec:asymptotic} we prove that asymptotically the Leiden algorithm guarantees the last two properties of uniform $\gamma$-density and subset optimality.

\subsection{Guarantees in each iteration}
\label{sec:each_iteration}

In order to show that the property of $\gamma$-separation is guaranteed in each iteration of the Leiden algorithm, we first need to prove some results for the \textsc{MoveNodesFast} function in the Leiden algorithm.

We start by introducing some notation.
The \textsc{MoveNodesFast} function iteratively evaluates nodes.
When a node is evaluated, either it is moved to a different (possibly empty) community or it is kept in its current community, depending on what is most beneficial for the quality function.
Let $G = (V, E)$ be a graph, let $\P$ be a partition of $G$, and let $\P' = \textsc{MoveNodesFast}(G, \P)$.
We denote by $\P_0, \ldots, \P_r$ a sequence of partitions generated by the \textsc{MoveNodesFast} function, with $\P_0 = \P$ denoting the initial partition, $\P_1$ denoting the partition after the first evaluation of a node has taken place, and so on.
$\P_r = \P'$ denotes the partition after the final evaluation of a node has taken place.
The \textsc{MoveNodesFast} function maintains a queue of nodes that still need to be evaluated.
Let $Q_s$ be the set of nodes that still need to be evaluated after $s$ node evaluations have taken place, with $Q_0 = V$.
Also, for all $v \in V$, let $C_s^v \in \P_s$ be the community in which node $v$ finds itself after $s$ node evaluations have taken place.

The following lemma states that at any point in the \textsc{MoveNodesFast} function, if a node is disconnected from the rest of its community, the node will find itself in the queue of nodes that still need to be evaluated.
\begin{lemma}
  Using the notation introduced above, for all $v \in V$ and all $s$, we have $v \in Q_s$ or $|C_s^v| = 1$ or $E(v, C_s^v - v) > 0$.
  \label{lem:move_nodes_fast1}
\end{lemma}
\begin{proof}
  We are going to prove the lemma for an arbitrary node $v \in V$.
  We provide a proof by induction.
  We observe that $v \in Q_0$, which provides our inductive base.
  Suppose that $v \in Q_{s - 1}$ or $|C_{s - 1}^v| = 1$ or $E(v, C_{s - 1}^v - v) > 0$.
  This is our inductive hypothesis.
  We are going to show that $v \in Q_s$ or $|C_s^v| = 1$ or $E(v, C_s^v - v) > 0$.
  If $v \in Q_s$, this result is obtained in a trivial way.
  Suppose therefore that $v \notin Q_s$.
  We then need to show that $|C_s^v| = 1$ or $E(v, C_s^v - v) > 0$.
  To do so, we distinguish between two cases.

  We first consider the case in which $v \in Q_{s - 1}$.
  If $v \in Q_{s - 1}$ and $v \notin Q_s$, node $v$ has just been evaluated.
  We then obviously have $|C_s^v| = 1$ or $E(v, C_s^v - v) > 0$.
  Otherwise we would have $|C_s^v| > 1$ and $E(v, C_s^v - v) = 0$, which would mean that node $v$ is disconnected from the rest of its community.
  Since node $v$ has just been evaluated, this is not possible.

  We now consider the case in which $v \notin Q_{s - 1}$.
  Let $u \in V$ be the node that has just been evaluated, i.e., $u \in Q_{s - 1}$ and $u \notin Q_s$.
  If node $u$ has not been moved to a different community, then $\P_s = \P_{s - 1}$.
  Obviously, if $|C_{s - 1}^v| = 1$ or $E(v, C_{s - 1}^v - v) > 0$, we then have $|C_s^v| = 1$ or $E(v, C_s^v - v) > 0$.
  On the other hand, if node $u$ has been moved to a different community, we have $(u, v) \notin E(G)$ or $v \in C_s^u$.
  To see this, note that if $(u, v) \in E(G)$ and $v \notin C_s^u$, we would have $v \in Q_s$ (following line~\ref{algo:leiden:add_neighs_to_queue} in Algorithm~\ref{algo:leiden}).
  This contradicts our assumption that $v \notin Q_s$, so that we must have $(u, v) \notin E(G)$ or $v \in C_s^u$.
  In other words, either there is no edge between nodes $u$ and $v$ or node $u$ has been moved to the community of node $v$.
  In either case, it is not possible that the movement of node $u$ causes node $v$ to become disconnected from the rest of its community.
  Hence, in either case, if $|C_{s - 1}^v| = 1$ or $E(v, C_{s - 1}^v - v) > 0$, then $|C_s^v| = 1$ or $E(v, C_s^v - v) > 0$.
\end{proof}

Using Lemma~\ref{lem:move_nodes_fast1}, we now prove the following lemma, which states that for partitions provided by the \textsc{MoveNodesFast} function it is guaranteed that singleton communities cannot be merged with each other.
\begin{lemma}
  Let $G = (V, E)$ be a graph, let $\P$ be a partition of $G$, and let $\P' = \textsc{MoveNodesFast}(G, \P)$.
  Then for all pairs $C, D \in \P'$ such that $|C| = |D| = 1$, we have $\Delta \Hf(C \mapsto D) = \Delta \Hf(D \mapsto C) \leq 0$.
  \label{lem:move_nodes_fast2}
\end{lemma}
\begin{proof}
  We are going to prove the lemma for an arbitrary pair of communities $C, D \in \P'$ such that $|C| = |D| = 1$.
  We use the notation introduced above.
  If $C, D \in \P_s$ for all $s$, it is clear that $\Delta \Hf(C \mapsto D) = \Delta \Hf(D \mapsto C) \leq 0$.
  Otherwise, consider $t$ such that $C, D \in \P_s$ for all $s \geq t$ and either $C \notin \P_{t - 1}$ or $D \notin \P_{t - 1}$.
  Without loss of generality, we assume that $C \notin \P_{t - 1}$ and $D \in \P_{t - 1}$.
  Consider $v \in V$ such that $C = \{v\}$.
  After $t - 1$ node evaluations have taken place, there are two possibilities.

  One possibility is that node $v$ is evaluated and is moved to an empty community.
  This means that moving node $v$ to an empty community is more beneficial for the quality function than moving node $v$ to community $D$.
  It is then clear that $\Delta \Hf(C \mapsto D) = \Delta \Hf(D \mapsto C) \leq 0$.

  The second possibility is that node $v$ is in a community together with one other node $u \in V$ (i.e. $\{u, v\} \in \P_{t - 1}$) and that this node $u$ is evaluated and is moved to a different community.
  In this case, $v \in Q_t$, as we will now show.
  If $(u, v) \in E(G)$, this follows from line~\ref{algo:leiden:add_neighs_to_queue} in Algorithm~\ref{algo:leiden}.
  If $(u, v) \notin E(G)$, we have $|C_{t - 1}^v| = |\{u, v\}| = 2$ and $E(v, C_{t - 1}^v - v) = 0$.
  It then follows from Lemma~\ref{lem:move_nodes_fast1} that $v \in Q_{t - 1}$.
  Since node $v$ is not evaluated in node evaluation $t$ (node $u$ is evaluated in this node evaluation), $v \in Q_{t - 1}$ implies that $v \in Q_t$.
  If $v \in Q_t$, at some point $s \geq t$, node $v$ is evaluated.
  Since $C, D \in \P_s$ for all $s \geq t$, keeping node $v$ in its own singleton community $C$ is more beneficial for the quality function than moving node $v$ to community $D$.
  This means that $\Delta \Hf(C \mapsto D) = \Delta \Hf(D \mapsto C) \leq 0$.
\end{proof}

Lemma~\ref{lem:move_nodes_fast2} enables us to prove that the property of $\gamma$-separation is guaranteed in each iteration of the Leiden algorithm, as stated in the following theorem.
\begin{theorem}
  Let $G = (V, E)$ be a graph, let $\P_t$ be a flat partition of $G$, and let $\P_{t + 1} = \textsc{Leiden}(G, \P_t)$.
  Then $\P_{t + 1}$ is $\gamma$-separated.
  \label{thm:gamma_separation}
\end{theorem}
\begin{proof}
  Let $G_\ell = (V_\ell, E_\ell)$ be the aggregate graph at the highest level in the Leiden algorithm, let $\P_\ell$ be the initial partition of $G_\ell$, and let $\P'_\ell = \textsc{MoveNodesFast}(G_\ell, \P_\ell)$.
  Since we are at the highest level of aggregation, it follows from line~\ref{algo:leiden:stop_criterion} in Algorithm~\ref{algo:leiden} that $|\P'_\ell| = |V_\ell|$, which means that $|C| = 1$ for all $C \in \P'_\ell$.
  In other words, $\P'_\ell$ is a singleton partition of $G_\ell$.
  Lemma~\ref{lem:move_nodes_fast2} then implies that for all $C, D \in \P'_\ell$ we have $\Delta \Hf(C \mapsto D) = \Delta \Hf(D \mapsto C) \leq 0$.
  Since $\P_{t + 1} = \flatf^*(\P'_\ell)$, it follows that for all $C, D \in \P_{t + 1}$ we have $\Delta \Hf(C \mapsto D) = \Delta \Hf(D \mapsto C) \leq 0$.
  Hence, $\P_{t + 1}$ is $\gamma$-separated.
\end{proof}

The property of $\gamma$-separation also holds after each iteration of the Louvain algorithm.
In fact, for the Louvain algorithm this is much easier to see than for the Leiden algorithm.
The Louvain algorithm uses the \textsc{MoveNodes} function instead of the \textsc{MoveNodesFast} function.
Unlike the \textsc{MoveNodesFast} function, the \textsc{MoveNodes} function yields partitions that are guaranteed to be node optimal.
This guarantee leads in a straightforward way to the property of $\gamma$-separation for partitions obtained in each iteration of the Louvain algorithm.

We now consider the property of $\gamma$-connectivity.
By constructing a tree corresponding to the decomposition of $\gamma$-connectivity, we are going to prove that this property is guaranteed in each iteration of the Leiden algorithm.
\begin{theorem}
  Let $G = (V, E)$ be a graph, let $\P_t$ be a flat partition of $G$, and let $\P_{t + 1} = \textsc{Leiden}(G, \P_t)$.
  Then $\P_{t + 1}$ is $\gamma$-connected.
  \label{thm:gamma_connectedness}
\end{theorem}
\begin{proof}
  Let $G_\ell = (V_\ell, E_\ell)$ be the aggregate graph at level $\ell$ in the Leiden algorithm, with $G_0 = G$ being the base graph.
  We say that a node $v \in V_\ell$ is $\gamma$-connected if $\flatf(v)$ is $\gamma$-connected.
  We are going to proceed inductively.
  Each node in the base graph $G_0$ is trivially $\gamma$-connected.
  This provides our inductive base.
  Suppose that each node $v \in V_{\ell - 1}$ is $\gamma$-connected, which is our inductive hypothesis.
  Each node $v \in V_\ell$ is obtained by merging one or more nodes at the preceding level, i.e. $v = \{u \mid u \in S\}$ for some set $S \subseteq V_{\ell - 1}$.
  If $v$ consists of only one node at the preceding level, $v$ is immediately $\gamma$-connected by our inductive hypothesis.
  The set of nodes $S$ is constructed in the \textsc{MergeNodesSubset} function.
  There exists some order $u_1, \ldots, u_k$ in which nodes are added to $S$.
  Let $S_i = \{u_1, \ldots, u_i\}$ be the set obtained after adding node $u_i$.
  It follows from line~\ref{algo:leiden:merge_prob} in Algorithm~\ref{algo:leiden} that $E(u_{i + 1}, S_i) \geq \gamma \|u_{i + 1}\| \cdot \|S_i\|$ for $i = 1, \ldots, k - 1$.
  Taking into account that each $u_i$ is $\gamma$-connected by our inductive hypothesis, this implies that each set $S_i$ is $\gamma$-connected.
  Since $S = S_k$ is $\gamma$-connected, node $v$ is $\gamma$-connected.
  Hence, each node $v \in V_\ell$ is $\gamma$-connected.
  This also holds for the nodes in the aggregate graph at the highest level in the Leiden algorithm, which implies that all communities in $\P_{t + 1}$ are $\gamma$-connected.
  In other words, $\P_{t + 1}$ is $\gamma$-connected.
\end{proof}
Note that the theorem does not require $\P_t$ to be connected.
Even if a disconnected partition is provided as input to the Leiden algorithm, performing a single iteration of the algorithm will give a partition that is $\gamma$-connected.

\subsection{Guarantees in stable iterations}
\label{sec:stable_iteration}

As discussed earlier, the Leiden algorithm can be iterated until $\P_{t + 1} = \textsc{Leiden}(G, \P_t)$.
Likewise, the Louvain algorithm can be iterated until $\P_{t + 1} = \textsc{Louvain}(G, \P_t)$.
We say that an iteration is \emph{stable} if $\P_{t + 1} = \P_t$, in which case we call $\P_t$ (or $\P_{t + 1}$) a \emph{stable partition}.

There is a subtle point when considering stable iterations.
In order for the below guarantees to hold, we need to ensure that $\Hf(\P_{t + 1}) = \Hf(\P_t)$ implies $\P_{t + 1} = \P_t$.
In both the Leiden algorithm and the Louvain algorithm, we therefore consider only strictly positive improvements (see line~\ref{algo:louvain:strict_increase_move} in Algorithm~\ref{algo:louvain} and line~\ref{algo:leiden:strict_increase_move} in Algorithm~\ref{algo:leiden}).
In other words, if a node movement leads to a partition that has the same quality as the current partition, the current partition is preferred and the node movement will not take place.
This then also implies that $\Hf(\P_{t + 1}) > \Hf(\P_t)$ if $\P_{t + 1} \neq \P_t$.

The Leiden algorithm guarantees that a stable partition is subpartition $\gamma$-dense, as stated in the following theorem. Note that the proof of the theorem has a structure that is similar to the structure of the proof of Theorem~\ref{thm:gamma_connectedness} presented above.
\begin{theorem}
  Let $G = (V, E)$ be a graph, let $\P_t$ be a flat partition of $G$, and let $\P_{t + 1} = \textsc{Leiden}(G, \P_t)$.
  If $\P_{t + 1} = \P_t$, then $\P_{t + 1} = \P_t$ is subpartition $\gamma$-dense.
  \label{thm:subpartition_gamma_density}
\end{theorem}
\begin{proof}
  Suppose we have a stable iteration.
  Hence, $\P_{t + 1} = \P_t$.
  Let $G_\ell = (V_\ell, E_\ell)$ be the aggregate graph at level $\ell$ in the Leiden algorithm, with $G_0 = G$ being the base graph.
  We say that a node $v \in V_\ell$ is subpartition $\gamma$-dense if the set of nodes $\flatf(v)$ is subpartition $\gamma$-dense.
  We first observe that for all levels $\ell$ and all nodes $v \in V_\ell$ we have $\Delta \Hf(v \mapsto \emptyset) \leq 0$.
  To see this, note that if $\Delta \Hf(v \mapsto \emptyset) > 0$ for some level $\ell$ and some node $v \in V_\ell$, the \textsc{MoveNodesFast} function would have removed node $v$ from its community, which means that the iteration would not have been stable.
  We are now going to proceed inductively.
  Since $\Delta \Hf(v \mapsto \emptyset) \leq 0$ for all nodes $v \in V_0$, each node in the base graph $G_0$ is subpartition $\gamma$-dense.
  This provides our inductive base.
  Suppose that each node $v \in V_{\ell - 1}$ is subpartition $\gamma$-dense, which is our inductive hypothesis.
  Each node $v \in V_\ell$ is obtained by merging one or more nodes at the preceding level, i.e. $v = \{u \mid u \in S\}$ for some set $S \subseteq V_{\ell - 1}$.
  If $v$ consists of only one node at the preceding level, $v$ is immediately subpartition $\gamma$-dense by our inductive hypothesis.
  The set of nodes $S$ is constructed in the \textsc{MergeNodesSubset} function.
  There exists some order $u_1, \ldots, u_k$ in which nodes are added to $S$.
  Let $S_i = \{u_1, \ldots, u_i\}$ be the set obtained after adding node $u_i$.
  It follows from line~\ref{algo:leiden:merge_prob} in Algorithm~\ref{algo:leiden} that $E(u_{i + 1}, S_i) \geq \gamma \|u_{i + 1}\| \cdot \|S_i\|$ for $i = 1, \ldots, k - 1$.
  Furthermore, line~\ref{algo:leiden:strict_merge2} in Algorithm~\ref{algo:leiden} ensures that $\Delta \Hf(S_i \mapsto \emptyset) \leq 0$ for $i = 1, \ldots, k - 1$.
  We also have $\Delta \Hf(S_k \mapsto \emptyset) \leq 0$, since $S_k = S = v$ and since $\Delta \Hf(v \mapsto \emptyset) \leq 0$, as observed above.
  Taking into account that each $u_i$ is subpartition $\gamma$-dense by our inductive hypothesis, this implies that each set $S_i$ is subpartition $\gamma$-dense.
  Since $S = S_k$ is subpartition $\gamma$-dense, node $v$ is subpartition $\gamma$-dense.
  Hence, each node $v \in V_\ell$ is subpartition $\gamma$-dense.
  This also holds for the nodes in the aggregate graph at the highest level in the Leiden algorithm, which implies that all communities in $\P_{t + 1} = \P_t$ are subpartition $\gamma$-dense.
  In other words, $\P_{t + 1} = \P_t$ is subpartition $\gamma$-dense.
\end{proof}

Subpartition $\gamma$-density does not imply node optimality.
It guarantees only that $\Delta \Hf(v \mapsto \emptyset) \leq 0$ for all $v \in V$, not that $\Delta \Hf(v \mapsto D) \leq 0$ for all $v \in V$ and all $D \in \P$.
However, it is easy to see that all nodes are locally optimally assigned in a stable iteration of the Leiden algorithm.
This is stated in the following theorem.
\begin{theorem}
  Let $G = (V, E)$ be a graph, let $\P_t$ be a flat partition of $G$, and let $\P_{t + 1} = \textsc{Leiden}(G, \P_t)$.
  If $\P_{t + 1} = \P_t$, then $\P_{t + 1} = \P_t$ is node optimal.
  \label{thm:node_optimality}
\end{theorem}
\begin{proof}
  Suppose we have a stable iteration.
  Hence, $\P_{t + 1} = \P_t$.
  We are going to give a proof by contradiction.
  Assume that $\P_{t + 1} = \P_t$ is not node optimal.
  There then exists a node $v \in C \in \P_t$ and a community $D \in \P_t$ (or $D = \emptyset$) such that $\Delta \Hf(v \mapsto D) > 0$.
  The \textsc{MoveNodesFast} function then moves node $v$ to community $D$.
  This means that $\P_{t + 1} \neq \P_t$ and that the iteration is not stable.
  We now have a contradiction, which implies that the assumption of $\P_{t + 1} = \P_t$ not being node optimal must be false.
  Hence, $\P_{t + 1} = \P_t$ is node optimal.
\end{proof}
In the same way, it is straightforward to see that the Louvain algorithm also guarantees node optimality in a stable iteration.

When the Louvain algorithm reaches a stable iteration, the partition is $\gamma$-separated and node optimal.
Since the Louvain algorithm considers only moving nodes and merging communities, additional iterations of the algorithm will not lead to further improvements of the partition.
Hence, in the case of the Louvain algorithm, if $\P_{t + 1} = \P_t$, then $\P_\tau = \P_t$ for all $\tau \geq t$.
In other words, when the Louvain algorithm reaches a stable iteration, all future iterations will be stable as well.
This contrasts with the Leiden algorithm, which may continue to improve a partition after a stable iteration.
We consider this in more detail below.

\subsection{Asymptotic guarantees}
\label{sec:asymptotic}

When an iteration of the Leiden algorithm is stable, this does not imply that the next iteration will also be stable.
Because of randomness in the refinement phase of the Leiden algorithm, a partition that is stable in one iteration may be improved in the next iteration.
However, at some point, a partition will be obtained for which the Leiden algorithm is unable to make any further improvements.
We call this an asymptotically stable partition.
Below, we prove that an asymptotically stable partition is uniformly $\gamma$-dense and subset optimal.

We first need to show what it means to define asymptotic properties for the Leiden algorithm.
The Leiden algorithm considers moving a node to a different community only if this results in a strict increase in the quality function.
As stated in the following lemma, this ensures that at some point the Leiden algorithm will find a partition for which it can make no further improvements.
\begin{lemma}
  Let $G = (V, E)$ be a graph, and let $\P_{t + 1} = \textsc{Leiden}(G, \P_t)$.
  There exists a $\tau$ such that $\P_t = \P_\tau$ for all $t \geq \tau$.
  \label{lem:convergence}
\end{lemma}
\begin{proof}
  Only strict improvements can be made in the Leiden algorithm.
  Consequently, if $\P_{t + 1} \neq \P_t$, then $\P_{t + 1} \neq \P_{t'}$ for all $t' \leq t$.
  Assume that there does not exist a $\tau$ such that $\P_t = \P_\tau$ for all $t \geq \tau$.
  Then for any $\tau$ there exists a $t > \tau$ such that $\P_t \neq \P_{t'}$ for all $t' < t$.
  This implies that the number of unique elements in the sequence $\P_0, \P_1, \ldots$ is infinite.
  However, this is not possible, because the number of partitions of $G$ is finite.
  Hence, the assumption that there does not exist a $\tau$ such that $\P_t = \P_\tau$ for all $t \geq \tau$ is false.
\end{proof}

According to the above lemma, the Leiden algorithm progresses towards a partition for which no further improvements can be made.
We can therefore define the notion of an asymptotically stable partition.
\begin{definition}
  Let $G = (V, E)$ be a graph, and let $\P_{t + 1} = \textsc{Leiden}(G, \P_t)$.
  We call $\P_\tau$ \emph{asymptotically stable} if $\P_t = \P_\tau$ for all $t \geq \tau$.
\end{definition}

We also need to define the notion of a minimal non-optimal subset.
\begin{definition}
  Let $G = (V, E)$ be a graph, and let $\P$ be a partition of $G$.
  A set $S \subseteq C \in \P$ is called a \emph{non-optimal subset} if $\Delta \Hf(S \mapsto D) > 0$ for some $D \in \P$ or for $D = \emptyset$.
  A set $S \subseteq C \in \P$ is called a \emph{minimal non-optimal subset} if $S$ is a non-optimal subset and if there does not exist a non-optimal subset $S' \subset S$.
  \label{def:minimal_non_optimal_subset}
\end{definition}

The following lemma states an important property of minimal non-optimal subsets.
\begin{lemma}
  Let $G = (V, E)$ be a graph, let $\P$ be a partition of $G$, and let $S \subseteq C \in \P$ be a minimal non-optimal subset.
  Then $\{S\}$ is an optimal partition of the subgraph induced by $S$.
  \label{lem:minimal_non_optimal_subset}
\end{lemma}
\begin{proof}
  Assume that $\{S\}$ is not an optimal partition of the subgraph induced by $S$.
  There then exists a set $S_1 \in S$ such that
  \begin{equation}
    E(S_1, S_2) - \gamma \|S_1\| \cdot \|S_2\| < 0,
    \label{eq:split}
  \end{equation}
  where $S_2 = S - S_1$.
  Let $D \in \P$ or $D = \emptyset$ such that $\Delta \Hf(S \rightarrow D) > 0$.
  Hence,
  \begin{equation}
    E(S, D) - \gamma \|S\| \cdot \|D\| > E(S, C - S) - \gamma \|S\| \cdot \|C - S\|.
    \label{eq:ineq_S}
  \end{equation}
  Because $S$ is a minimal non-optimal subset, $S_1$ and $S_2$ cannot be non-optimal subsets.
  Therefore, $\Delta \Hf(S_1 \rightarrow D) \leq 0$ and $\Delta \Hf(S_2 \rightarrow D) \leq 0$, or equivalently,
  \begin{align}
    E(S_1, D) - \gamma \|S_1\| \cdot \|D\| &\leq E(S_1, C - S_1) - \gamma \|S_1\| \cdot \|C - S_1\|
    \label{eq:ineq_S_1}
  \intertext{and}
    E(S_2, D) - \gamma \|S_2\| \cdot \|D\| &\leq E(S_2, C - S_2) - \gamma \|S_2\| \cdot \|C - S_2\|.
    \label{eq:ineq_S_2}
  \end{align}
  It then follows from Eqs.~(\ref{eq:ineq_S_1}) and~(\ref{eq:ineq_S_2}) that
  \begin{align*}
    E(S, D) - \gamma \|S\| \cdot \|D\|
      &= \bigl(E(S_1, D) - \gamma \|S_1\| \cdot \|D\|\bigr) + \bigl(E(S_2, D) - \gamma \|S_2\| \cdot \|D\|\bigr) \\
      &\leq \bigl(E(S_1, C - S_1) - \gamma \|S_1\| \cdot \|C - S_1\|\bigr) + \bigl(E(S_2, C - S_2) - \gamma \|S_2\| \cdot \|C - S_2\|\bigr).
  \end{align*}
  This can be written as
  \begin{alignat*}{3}
    E(S, D) - \gamma \|S\| \cdot \|D\|
    \leq & & \bigl(E(S_1, C - S) + E(S_1, S_2) - \gamma \|S_1\| \cdot \|C - S\| - \gamma \|S_1\| \cdot \|S_2\|\bigr) \\
     &+& \bigl(E(S_2, C - S) + E(S_2, S_1) - \gamma \|S_2\| \cdot \|C - S\| - \gamma \|S_2\| \cdot \|S_1\|\bigr) \\
    =& & E(S, C - S) + 2E(S_1, S_2) - \gamma \|S\| \cdot \|C - S\| - 2\gamma \|S_1\| \cdot \|S_2\|.
  \end{alignat*}
  Using Eq.~(\ref{eq:split}), we then obtain
  \begin{equation*}
    E(S, D) - \gamma \|S\| \cdot \|D\| < E(S, C - S) - \gamma \|S\| \cdot \|C - S\|.
  \end{equation*}
  However, this contradicts Eq.~(\ref{eq:ineq_S}).
  The assumption that $\{S\}$ is not an optimal partition of the subgraph induced by $S$ is therefore false.
\end{proof}

Building on the results for non-decreasing move sequences reported in Appendix~\ref{sec:nondecreasing_move_sequences}, the following lemma states that any minimal non-optimal subset can be found by the \textsc{MergeNodeSubset} function.
\begin{lemma}
  Let $G = (V, E)$ be a graph, let $\P$ be a partition of $G$, and let $S \subseteq C \in \P$ be a minimal non-optimal subset.
  Let $\P_\text{refined} = \textsc{MergeNodesSubset}(G, \{\{v\} \mid v \in V\}, C)$.
  There then exists a move sequence in the \textsc{MergeNodesSubset} function such that $S \in \P_\text{refined}$.
  \label{lem:find_optimal}
\end{lemma}
\begin{proof}
  We are going to prove that there exists a move sequence $\P_0, \ldots, \P_{|C|}$ in the \textsc{MergeNodesSubset} function such that $S \in \P_{|C|}$.
  The move sequence consists of two parts, $\P_0, \ldots, \P_{|S|}$ and $\P_{|S|}, \ldots, \P_{|C|}$.
  In the first part, each node in $S$ is considered for moving.
  In the second part, each node in $C - S$ is considered for moving.
  Note that in the \textsc{MergeNodesSubset} function a node can always stay in its own community when it is considered for moving.
  We first consider the first part of the move sequence $\P_0, \ldots, \P_{|C|}$.
  Let $\P_0, \ldots, \P_{|S|}$ be a non-decreasing move sequence such that $\P_0 = \{\{v\} \mid v \in V\}$ and $S \in \P_{|S|}$.
  To see that such a non-decreasing move sequence exists, note that according to Lemma~\ref{lem:minimal_non_optimal_subset} $\{S\}$ is an optimal partition of the subgraph induced by $S$ and that according to Theorem~\ref{thm:optimal_seq} an optimal partition can be reached using a non-decreasing move sequence.
  This non-decreasing move sequence consists of $|S| - 1$ moves.
  There is one node in $S$ that can stay in its own community.
  Note further that each move in the move sequence $\P_0, \ldots, \P_{|S|}$ satisfies the conditions specified in lines~\ref{algo:leiden:strict_merge1} and~\ref{algo:leiden:strict_merge2} in Algorithm~\ref{algo:leiden}.
  This follows from Definition~\ref{def:minimal_non_optimal_subset}.
  In the second part of the move sequence $\P_0, \ldots, \P_{|C|}$, we simply have $\P_{|S|} = \ldots = \P_{|C|}$.
  Hence, each node in $C - S$ stays in its own community.
  Since $S \in \P_{|S|}$, we then also have $S \in \P_{|C|}$.
\end{proof}

As long as there are subsets of communities that are not optimally assigned, the \textsc{MergeNodesSubset} function can find these subsets.
In the \textsc{MoveNodesFast} function, these subsets are then moved to a different community.
In this way, the Leiden algorithm continues to identify better partitions.
However, at some point, all subsets of communities are optimally assigned, and the Leiden algorithm will not be able to further improve the partition.
The algorithm has then reached an asymptotically stable partition, and this partition is also subset optimal.
This result is formalized in the following theorem.
\begin{theorem}
  Let $G = (V, E)$ be a graph, and let $\P$ be a flat partition of $G$.
  Then $\P$ is asymptotically stable if and only if $\P$ is subset optimal.
  \label{thm:subset_optimality}
\end{theorem}
\begin{proof}
  If $\P$ is subset optimal, it follows directly from the definition of the Leiden algorithm that $\P$ is asymptotically stable.
  Conversely, if $\P$ is asymptotically stable, it follows from Lemma~\ref{lem:find_optimal} that $\P$ is subset optimal.
  To see this, assume that $\P$ is not subset optimal.
  There then exists a community $C \in \P$ and a set $S \subset C$ such that $S$ is a minimal non-optimal subset.
  Let $\P_\text{refined} = \textsc{MergeNodesSubset}(G, \{\{v\} \mid v \in V\}, C)$.
  Lemma~\ref{lem:find_optimal} states that there exists a move sequence in the \textsc{MergeNodesSubset} function such that $S \in \P_\text{refined}$.
  If $S \in \P_\text{refined}$, then $S$ will be moved from $C$ to a different (possibly empty) community in line~\ref{algo:leiden:call_move_nodes} in Algorithm~\ref{algo:leiden}.
  However, this contradicts the asymptotic stability of $\P$.
  Asymptotic stability therefore implies subset optimality.
\end{proof}

Since subset optimality implies uniform $\gamma$-density, we obtain the following corollary.
\begin{corollary}
  Let $G = (V, E)$ be a graph, and let $\P$ be a flat partition of $G$.
  If $\P$ is asymptotically stable, then $\P$ is uniformly $\gamma$-dense.
\end{corollary}

\section{Bounds on optimality}
\label{sec:bound_on_optimality}

In this appendix, we prove that the quality of a uniformly $\gamma$-dense partition as defined in Definition~\ref{def:uniform_gamma_density} in Appendix~\ref{sec:guarantees} provides an upper bound on the quality of an optimal partition.

We first define the intersection of two partitions.
\begin{definition}
  Let $G = (V, E)$ be a graph, and let $\P_1$ and $\P_2$ be flat partitions of $G$.
  We denote the \emph{intersection} of $\P_1$ and $\P_2$ by $\P = \P_1 \sqcap \P_2$, which is defined as
  \begin{equation}
    \P = \{ C \cap D \mid C \in \P_1, D \in \P_2, C \cap D \neq \emptyset \}.
  \end{equation}
\end{definition}
The intersection of two partitions consists of the basic subsets that form both partitions.
For $S, R \in \P = \P_1 \sqcap \P_2$, we write $S \simas{\P_1} R$ if there exists a community $C \in \P_1$ such that $S, R \subseteq C$.
Hence, if $S \simas{\P_1} R$, then $S$ and $R$ are subsets of the same community in $\P_1$.
Furthermore, for $S \neq R$, if $S \simas{\P_1} R$, then we cannot have $S \simas{\P_2} R$, since otherwise $S$ and $R$ would have formed a single subset.
In other words, $S \simas{\P_1} R \Rightarrow S \nsimas{\P_2} R$ and similarly $S \simas{\P_2} R \Rightarrow S \nsimas{\P_1} R$.

The following lemma shows how the difference in quality between two partitions can easily be expressed using the intersection.
\begin{lemma}
  Let $G = (V, E)$ be a graph, let $\P_1$ and $\P_2$ be flat partitions of $G$, and let $\P = \P_1 \sqcap \P_2$ be the intersection of $\P_1$ and $\P_2$.
  Then
  \begin{equation}
    \Hf(\P_2) - \Hf(\P_1) =
    \frac{1}{2} \sum_{\substack{S \simas{\P_2} R \\ S \neq R}} \bigl[E(S, R) - \gamma \|S\| \cdot \|R\|\bigr]
    -
    \frac{1}{2} \sum_{\substack{S \simas{\P_1} R \\ S \neq R}} \bigl[E(S, R) - \gamma \|S\| \cdot \|R\|\bigr].
    \label{eq:diff_partition}
  \end{equation}
  \label{lem:diff_partition}
\end{lemma}
\begin{proof}
  For any community $C \in \P_k$ ($k = 1, 2$),
  \begin{equation*}
    E(C, C) = \sum_{\substack{S \in \P \\ S \subseteq C}} E(S, S) + \frac{1}{2}\sum_{\substack{S, R \in \P \\ S, R \subseteq C \\ S \neq R}} E(S, R)
  \end{equation*}
  and
  \begin{equation*}
  \binom{\|C\|}{2} = \sum_{\substack{S \in \P \\ S \subseteq C}} \binom{\|S\|}{2} + \frac{1}{2} \sum_{\substack{S, R \in \P \\ S, R \subseteq C \\ S \neq R}} \|S\| \cdot \|R\|.
  \end{equation*}
  We hence obtain
  \begin{align*}
    \Hf(\P_k) =& \sum_{C \in \P_k} \left[E(C, C) - \gamma \binom{\|C\|}{2} \right] \\
    =& \sum_{C \in \P_k} \left[\sum_{\substack{S \in \P \\ S \subseteq C}} \left[E(S, S) - \gamma \binom{\|S\|}{2}\right] +
       \frac{1}{2}\sum_{\substack{S, R \in \P \\ S, R \subseteq C \\ S \neq R}} \left[E(S, R) - \gamma \|S\| \cdot \|R\| \right] \right] \\
    =& \sum_{\substack{S \in \P}} \left[E(S, S) - \gamma \binom{\|S\|}{2}\right] +
       \frac{1}{2}\sum_{\substack{S \simas{\P_k} R \\ S \neq R}} \left[E(S, R) - \gamma \|S\| \cdot \|R\| \right].
  \end{align*}
  The difference $\Hf(\P_2) - \Hf(\P_1)$ then gives the desired result.
\end{proof}

The above lemma enables us to prove the following theorem, stating that the quality of a uniformly $\gamma$-dense partition is not too far from optimal.
We stress that this theorem applies only to unweighted graphs.
\begin{theorem}
  Let $G = (V, E)$ be an unweighted graph, let $\P$ be a uniformly $\gamma$-dense partition of $G$, and let $\P^*$ be an optimal partition of $G$.
  Then
  \begin{equation}
    \Hf(\P^*) - \Hf(\P) \leq (1 - \gamma)\frac{1}{2}\sum_{\substack{C, D \in \P \\ C \neq D}} E(C, D).
    \label{eq:bound_external_edges}
  \end{equation}
  \label{thm:bound_external_edges}
\end{theorem}
\begin{proof}
  Let $\P' = \P \sqcap \P^*$.
  Consider any $S, R \in \P'$ such that $S \simas{\P^*} R$.
  Because the graph $G$ is unweighted, we have $\|S\| \cdot \|R\| \geq E(S, R)$.
  It follows that
  \begin{equation*}
    E(S, R) - \gamma \|S\| \cdot \|R\|
        \leq (1 - \gamma)E(S, R).
  \end{equation*}
  Furthermore, for any community $C \in \P$, the number of edges connecting this community with other communities is $E(C, V - C)$.
  We therefore have
  \begin{equation*}
    \sum_{S \subseteq C} \sum_{\substack{R \simas{\P^*} S \\ R \neq S}} E(S, R)
        \leq E(C, V - C).
  \end{equation*}
  To see this, note that $R \simas{\P^*} S$ implies $R \nsimas{\P} S$, so that $R \not\subseteq C$.
  For any $C \in \P$, we then obtain
  \begin{equation*}
    \sum_{S \subseteq C} \sum_{\substack{R \simas{\P^*} S \\ R \neq S}} \left[E(S, R) - \gamma \|S\| \cdot \|R\|\right]
        \leq \sum_{S \subseteq C} \sum_{\substack{R \simas{\P^*} S \\ R \neq S}} (1 - \gamma)E(S, R)
        \leq (1 - \gamma)E(C, V - C).
  \end{equation*}
  By summing over all $C \in \P$, this gives
  \begin{equation*}
    \sum_{\substack{S \simas{\P^*} R \\ S \neq R}} \left[E(S, R) - \gamma \|S\| \cdot \|R\|\right]
	    \leq (1 - \gamma)\sum_{\substack{C, D \in \P \\ C \neq D}} E(C, D).
  \end{equation*}
  Furthermore, because $\P$ is uniformly $\gamma$-dense, we have
  \begin{equation*}
    \sum_{\substack{S \simas{P} R \\ S \neq R}} \left[E(S, R) - \gamma \|S\| \cdot \|R\|\right]
        \geq 0.
  \end{equation*}
  Using these results, Eq.~(\ref{eq:bound_external_edges}) follows from Lemma~\ref{lem:diff_partition}.
\end{proof}

For weighted graphs, an upper bound analogous to Eq.~(\ref{eq:bound_external_edges}) is
\begin{equation}
  \Hf(\P^*) - \Hf(\P) \leq \left(1 - \frac{\gamma}{\bar{w}}\right) \frac{1}{2} \sum_{C, D \in \P} E(C, D),
  \label{eq:bound_weighted}
\end{equation}
where $\bar{w} = \max_{i, j} w_{i, j}$ is the maximum edge weight.

For modularity instead of CPM, the upper bound for unweighted graphs in Eq.~(\ref{eq:bound_external_edges}) needs to be adjusted by rescaling the resolution parameter by $2m$.
This gives
\begin{equation}
  \Hf(\P^*) - \Hf(\P) \leq \left(1 - \frac{\gamma}{2m}\right) \frac{1}{2}\sum_{\substack{C, D \in \P \\ C \neq D}} E(C, D).
\end{equation}
The approximation factor of modularity cannot be multiplicative~\cite{Dinh2015}, and indeed our bound is additive.
Depending on the partition $\P$, our bound may be better than the bound provided by an SDP algorithm~\cite{Dinh2015}.

Note that the bound in Eq.~(\ref{eq:bound_external_edges}) reduces the trivial bound of $(1 - \gamma) m$ by $\gamma$ times the number of missing links within communities, i.e., $\gamma \sum_{C} \left[\binom{\|C\|}{2} - E(C, C)\right]$.
To see this, note that $m = \sum_{C} E(C, C) + \frac{1}{2} \sum_{C \neq D} E(C, D)$.
Starting from Eq.~(\ref{eq:bound_external_edges}), we then obtain
\begin{align*}
  \Hf(\P^*) &\leq \Hf(\P) + (1 - \gamma) \frac{1}{2} \sum_{\substack{C, D \in \P \\ C \neq D}} E(C, D) \\
  &= \sum_{C \in \P} \left[E(C, C) - \gamma \binom{\|C\|}{2}\right] + (1 - \gamma) m - (1 - \gamma) \sum_{C \in \P} E(C, C) \\
  &= (1 - \gamma) m - \gamma \sum_{C \in \P} \left[\binom{\|C\|}{2} - E(C, C)\right].
\end{align*}

Finally, Theorem~\ref{thm:bound_external_edges} provides a bound on the quality of the optimal partition for a given uniformly $\gamma$-dense partition, but it does not provide an a priori bound on the minimal quality of a uniformly $\gamma$-dense partition.
Finding such an a priori bound remains an open problem.

\end{document}